\documentclass[11pt,twoside]{article}

\usepackage{xcolor}

\usepackage{fullpage}

\usepackage{epsf}
\usepackage{fancyhdr}
\usepackage{graphics}
\usepackage{graphicx} 
\usepackage{float} 
\usepackage{subfigure} 
\usepackage{psfrag}
\usepackage{comment}
\usepackage{cancel}

\usepackage[linesnumbered,ruled]{algorithm2e}% http://ctan.org/pkg/algorithm2e
\DontPrintSemicolon	

\usepackage{color}
\usepackage{amsthm}
\usepackage{amsfonts}
\usepackage{amsmath}
\usepackage{bm}
\usepackage{amssymb,bbm}
\usepackage[numbers]{natbib}
\usepackage{algorithmic}
\usepackage[usestackEOL]{stackengine}
\usepackage{url}
\usepackage[colorlinks=True,linkcolor=magenta,citecolor=blue,urlcolor=blue,pagebackref=true,backref=page]
{hyperref}
\renewcommand*{\backref}[1]{\ifx#1\relax \else Page #1 \fi}
\renewcommand*{\backrefalt}[4]{%
    \ifcase #1 \footnotesize{(Not cited.)}%
    \or        \footnotesize{(Cited on page~#2.)}%
    \else      \footnotesize{(Cited on pages~#2.)}%
    \fi}

\usepackage{nicefrac}

\def\hatt{\widehat}

\usepackage{chngpage}

\usepackage{tabularx}%

\usepackage{enumitem}
\usepackage{booktabs}
\usepackage{pbox}

\usepackage{caption}

\usepackage{mathtools}

\usepackage{fullpage}
\allowdisplaybreaks
\newtheoremstyle{named}{}{}{\itshape}{}{\bfseries}{.}{.5em}{\thmnote{#3's }#1}
\theoremstyle{named}

\theoremstyle{plain}

\newtheorem{theorem}{Theorem}

\newtheorem{lemma}{Lemma}

\newtheorem{definition}{Definition}

\newlength{\widebarargwidth}
\newlength{\widebarargheight}
\newlength{\widebarargdepth}

\makeatletter
\long\def\@makecaption#1#2{
        \vskip 0.8ex
        \setbox\@tempboxa\hbox{\small {\bf #1:} #2}
        \parindent 1.5em  %% How can we use the global value of this???
        \dimen0=\hsize
        \advance\dimen0 by -3em
        \ifdim \wd\@tempboxa >\dimen0
                \hbox to \hsize{
                        \parindent 0em
                        \hfil
                        \parbox{\dimen0}{\def\baselinestretch{0.96}\small
                                {\bf #1.} #2
                                }
                        \hfil}
        \else \hbox to \hsize{\hfil \box\@tempboxa \hfil}
        \fi
        }
\makeatother

\long\def\comment#1{}
\definecolor{battleshipgrey}{rgb}{0.52, 0.52, 0.51}
\definecolor{darkgray}{rgb}{0.66, 0.66, 0.66}
\definecolor{darkgreen}{rgb}{0.0, 0.2, 0.13}
\definecolor{darkspringgreen}{rgb}{0.09, 0.45, 0.27}
\definecolor{dukeblue}{rgb}{0.0, 0.0, 0.61}
\definecolor{olivedrab7}{rgb}{0.24, 0.2, 0.12}
\definecolor{darkblue}{rgb}{0.0, 0.0, 0.55}
\definecolor{darkscarlet}{rgb}{0.34, 0.01, 0.1}
\definecolor{candyapplered}{rgb}{1.0, 0.03, 0.0}
\definecolor{ao(english)}{rgb}{0.0, 0.5, 0.0}
\definecolor{applegreen}{rgb}{0.55, 0.71, 0.0}

\SetKwInput{KwInput}{Input}                % Set the Input
\SetKwInput{KwOutput}{Output}              % set the Output

\newcommand{\expectation}{\mathbb{E}}

\newcommand{\expect}{\mathbb{E}}

\begin{document}

\title{Martingale Posterior Distributions for Log-concave Density Functions}
\author{Fuheng Cui \& Stephen G. Walker \\ \\
Department of Statistics and Data Sciences \\
The University of Texas at Austin\\
email: fuheng.cui@austin.utexas.edu, s.g.walker@math.utexas.edu
}

\date{}
\maketitle
\begin{abstract}
The family of log-concave density functions contains various kinds of common probability distributions. Due to the shape restriction, it is possible to find the nonparametric estimate of the density, for example, the nonparametric maximum likelihood estimate (NPMLE). However, the associated uncertainty quantification of the NPMLE is less well developed. The current techniques for uncertainty quantification are Bayesian, using a Dirichlet process prior combined with the use of Markov chain Monte Carlo (MCMC) to sample from the posterior. In this paper, we start with the NPMLE and use a version of the martingale posterior distribution to establish uncertainty about the NPMLE. The algorithm  can be implemented in  parallel and hence is fast. We prove the convergence of the algorithm by constructing suitable submartingales. We also illustrate results with different models and settings and some real data, and compare our method with that within the literature.
\end{abstract}


%\textsl{Keywords:} Conditional distribution function;  Kernel smoothing; Mixing distribution; Nonparametric estimator.
%\let\thefootnote\relax\footnotetext{$\star$ Note1. }
%\let\thefootnote\relax\footnotetext{$\ddag$ Note2.}
%\section{Introduction}

%\vspace{0.5 em}
%\noindent
%\textbf{Organization.}  

%\vspace{0.5 em}
%\noindent
%\textbf{Notation.} ...

\section{Introduction}
\label{sec:Introduction}
In this paper we present an approach to the Bayesian nonparametric analysis of a log-concave density function. Shape constrained nonparametric inference varies in the degree of complexity depending on the constraint. For example, for a monotone decreasing density on $(0,\infty)$ one has the representation
$$f(x)=\int_x^\infty s^{-1}\,dP(s)$$
for some distribution function $P$. A prior on $P$ would be easily available, such as a Dirichlet process. On the other hand, it is less straightforward when $f$ is log-concave. Mariucci et al.~\cite{Mariucci_2020} use the concave function
$$w(x)=\gamma_1\,\int_0^\infty s^{-1}\min(s,x)\,dP(s)-\gamma_2 x,$$
for $\gamma_1>0$ and some $\gamma_2$, with $f(x)$ taken to be proportional to $\exp w(x)$. If, for example, the prior for $P$ is taken to be a Dirichlet process then necessarily analysis will involve some complicated MCMC algorithms. Indeed, there is limited Bayesian nonparametric analysis of log-concave density functions within the literature, with most articles on the topic concentrating on theoretical aspects, such as asymptotics.

On the other hand, the nonparametric maximum likelihood estimator, which we write as NPMLE, is easily available. And there are a number of packages which compute the function; for example \texttt{logcondens}, by Rufibach and Dumbgen (2006).  

%\section{NPMLE of log-concave density}
%\label{sec:NPMLE_of_Log_concave_Distributions}

A comprehensive description of the NPMLE is provided in \cite{Dumbgen_2009}. For concreteness, we present the key results here. If $X_1<\cdots<X_N$ are the ordered data from a log-concave density $f(x)=\exp\,\phi(x)$, then the NPMLE is recovered as $\hatt\phi$ which exists uniquely and is piece-wise linear on each interval $[X_{j-1},X_j]$, $2\leq j\leq N$, with $\hatt\phi\equiv -\infty$ outside of $[X_1,X_N]$. One of the key properties for $\hatt\phi$ is that 
\begin{equation}\label{key1}
\int \Delta(x)\,dF_n(x)\leq \int\Delta(x)\,\exp\hatt\phi(x)\,dx
\end{equation}
for all $\Delta:\mathbb{R}\to\mathbb{R}$ for which
$\hatt\phi+\lambda \Delta$ is concave for some $\lambda>0$, where $F_n$ denotes the empirical distribution function. This result can be found in \cite{Dumbgen_2009}.

Our aim is to obtain a martingale posterior distribution for expressing uncertainty about an unknown log-concave density. We regard the foundations for the posterior to be that uncertainty arises from the missing data (i.e. what has not been seen) denoted by $X_{N+1:\infty}$. To deal with this, we construct a probability model for $X_{N+1:\infty}$ given $X_{1:N}$, as would be required. Just as the Bayesian bootstrap uses the empirical cumulative distribution function to generate random distributions from the posterior, via the P\'olya-urn scheme, so we use the NPMLE in a similar way, i.e. we sample $X_{N+1}\sim\widehat{F}_N$ and update $\widehat{F}_N$ to $\widehat{F}_{N+1}$ using $X_{1:N+1}$, and so on, for $N+2$ to infinity.

The layout of the paper is as follows. The martingale posterior distribution is described in Section~\ref{sec:The martingale posterior distribution}, in which the algorithm is presented for sampling from the posterior. We also present  theoretical properties, including convergence, showing that the posterior exists, and demonstrating asymptotic exchangeability for the sequence $X_{N+1:\infty}$. In Section~\ref{sec:Illustration}, we illustrate results with various models and settings and some real data, and compare our method with that in~\cite{Mariucci_2020}. Section~\ref{sec:Summary_and_Discussion} provides a brief discussion and summary.

\section{The martingale posterior distribution}
\label{sec:The martingale posterior distribution}
To find the NPMLE of the log-concave density function, consider the following functional optimization problem with the Lagrange multiplier $\lambda$, which is $-1$,
\begin{align}
  \label{equ:npmle_loss}
    L_N(\phi) = \int \phi(x) \,dF_N(x)+\lambda\,\left(\int \exp \phi(x)dx-1\right)
\end{align}
where $\phi(x)$ is a concave function. Define $\widehat{\phi}_N(x) = \underset{\phi \text { concave }}{\arg\max } L_{N}(\phi)$ and $\widehat{f}_N = \exp \hat{\phi}_N$ as the NPMLE. The existence and uniqueness of the NPMLE is studied, for example, in Theorem 2.1 in \cite{Dumbgen_2009}:
\begin{theorem}
\label{thm:existence_uniqueness}
    The NPMLE $\widehat{\phi}_N$ exists and is unique. It is linear on all intervals $[X_j,X_{j+1}]$, $1\leq j <n$. Moreover, $\widehat{\phi}_N = -\infty$ on $\mathbb{R}\backslash[X_1,X_N]$.
\end{theorem}

\noindent
We now present the algorithm for defining and sampling from the martingale posterior distribution and then we present the corresponding theory, in particular, that the posterior is well defined.

\subsection{Algorithm}
\label{subsec:Algorithm}
The Bayesian bootstrap (BB) was introduced by Rubin in 1981, see \cite{Rubin_1981}. It only uses the empirical distribution function, so the BB can be viewed as data-driven. The idea behind the BB is to provide uncertainty quantification which comes from the unseen data; see~\cite{Fong_2021}. The BB for the empirical distribution function can be expressed via a P\'olya-urn process: i.e. the update from $N$ to $N+1$ with sample $X_{N+1}$ is given by
$$
F_{N+1}(x)=\frac{N\,F_{N}(x)+1(X_{N+1}\leq x)}{N+1},
$$
where $F_N$ is the empirical distribution function given data $X_{1:N}$, $F_N(x)=N^{-1}\sum_{i=1}^N 1(X_i\leq x)$, and $X_{N+1}\sim F_N$. We use the same idea as that of the BB, which is using the data sampled from the current  distribution to update the distribution for the next step. The BB process, by continuing the above procedure for all $n>N+1$, creates a martingale for the $(F_n)_{n>N}$.

In our algorithm, we will have a submartingale instead of a martingale. Under some conditions, the convergent (sub)martingale posterior distribution can also be obtained from the NPMLE with NPMLE update, analogous to the BB from the empirical distribution with P\'olya-urn update.

The posterior distribution is based on the ideas appearing in \cite{Fong_2021}. The posterior is constructed by treating the uncertainty as arising from the missing observations $X_{N+1:\infty}$. 
%where we now represent the data in the unusual unordered way. 
A representation of the uncertainty is provided by the conditional density
$$p(X_{N+1:\infty}\mid X_{1:N})$$
and this in turn is based on a sequential construct of the form $p(X_n\mid X_{1:n-1})$ for all $n>N$.
The distribution for the next observation $X_n$ is taken from  the current NPMLE, i.e. $\hatt{F}_{n-1}$. As with \cite{Fong_2021}, the NPMLE is then updated using the new observation $X_n$. This proceeds to a large number $M$ ensuring convergence of $\hatt{F}_M$. Each $\hatt{F}_M$ obtained, by running such a strategy repeatedly, is a sample from the nonparametric posterior distribution. 

It is easy to choose $M$ so that $\widehat{F}_M$ has suitable convergence. Theoretically, we need to calculate $\sup \|\widehat{\phi}_n-\widehat{\phi}_{n+1}\|$ to select the stopping criterion $M$. The $\widehat{\phi}_n$ are all piece-wise linear functions, hence so is $\widehat{\phi}_{n+1}-\widehat{\phi}_{n}$. Therefore, $\sup \|\widehat{\phi}_{n+1}-\widehat{\phi}_{n}\|$ occurs at one of the points from the union of the non-differentiable points of both functions, i.e. the knot sets. Moreover, by Corollary 2.5 in~\cite{Dumbgen_2009}, $F_n-1/n\leq\widehat{F}_n\leq F_n$ on the set of knot points of $\widehat{\phi}_n$. Hence, according to the convergence of the distribution function on the knot set, a subset of the sample set $X_{1:n}$, the stopping criterion $M$ can be assigned based on the convergence on these sets of points.
%the initial data size and the tolerance error. In addition, for any $\phi$ such that $\int \exp(\phi)dx=1$, $L_n(\phi)=\sum_{I=1}^n \phi(X_i)$. Then $|L_{n+1}({\phi})-L_n(\phi)|$ is easier to compute than $\sup \|\widehat{\phi}_{n+1}-\widehat{\phi}_{n}\|$ sometimes. So the stop criterion $M$ can be set based on the convergence of $L_n(\phi)$ in practice. 
%Furthermore, if using a relaxed convergence criterion, which is the convergence of the means of the bootstrapped NPMLE distributions, the stop criterion can be determined by the convergence of the means of the empirical distributions, since both means are identical by Corollary 2.3 in~\cite{Dumbgen_2009}.

\begin{algorithm}[!ht]
\DontPrintSemicolon
  
  \KwInput{Observed data $x_{1:N}=(x_1,\ldots, x_N)$.}
  Compute the NPMLE of the data; i.e. $p_0(x)=\exp\{\phi(x;x_{1:N})\}$\\
  \For{$l= 1,\ldots,B$}{
  Set $p_{l,0}(x)=p_0(x)$\\
  \For{$n=N+1,\ldots,M$}{
  $X_{l,n} \sim p_{l,n-1}(x)$\\
  $p_{l,n}(x)=\exp\{\phi(x; X_{1:n})\}$
  }
  Set $p_l=p_{l,M}$
  }
  \KwOutput{$\{p_l\}_{l=1}^B$}
\caption{Functional Bootstrapping Update of Log-concave Distributions}
\label{algo:1}
\end{algorithm}

This process is detailed in Algorithm~\ref{algo:1} and next we demonstrate convergence of a sequence $(\hatt{F}_n)_{n>N}$ to the random distribution function $\hatt{F}_\infty$.

\subsection{Theory}
\label{subsec:Theory}

%\subsection{Convergence}
%\label{sec:Convergence}

The aim in this subsection is to show that the sequence of maximizers of (\ref{equ:npmle_loss}), i.e. $(\hatt\phi_m)$ for $m=N+1, N+2, \ldots$, converges almost surely to a random concave function $\phi_\infty$. Given a data set $\{X_1,\ldots,X_N\}$ with $X_1\leq X_2\leq\cdots\leq X_N$, let $\mathcal{C}$ be the set of all continuous concave functions on $[X_1,X_N]$. For any $n>N$, {define
$
L_n(\phi)=\int \phi\, dF_n - \int \exp(\phi)\, dx
$}
with the $\{X_{N+1}, X_{N+2}, \ldots\}$ arising as in Algorithm~\ref{algo:1}. 

\begin{lemma}[Almost sure convergence of the loss functions]
\label{lemma:Almost sure convergence of the loss functions}
For any $\phi\in\mathcal{C}$, $L_n(\phi)$ converges to some $L_{\infty}(\phi)$ almost surely.
\end{lemma}
\begin{proof}
    For any $\phi \in \mathcal{C}$, by Theorem 2.2 in~\cite{Dumbgen_2009}, it is that
    $$\int \phi\, dF_n\leq \int \phi \,d\widehat{F}_n\quad\mbox{a.s.}$$
    Since $X_{n+1}$ comes from $\widehat{F}_n$ for $n>N$ we have that
    $$E\,\left\{\int \phi\,dF_{n+1}\mid X_{1:n}\right\}
    =\frac{n}{n+1}\,\int \phi\,dF_n+\frac{1}{n+1}\int\phi \,d\widehat{F}_n.$$
Hence, $L_n(\phi)$ is a sub-martingale for all $\phi$. Since 
%   \begin{align*}
%        \mathbb{E}[L_{n+1}(\phi)|X_1,\cdots,X_n]&=\int \phi dF_n-\int \exp(\phi)dx\\
%        &= \frac{n}{n+1} L_n(\phi) + \mathbb{E}[\frac{1}{n+1}\phi(X_{n+1})|X_1,\cdots,X_n]-\frac{1}{n+1}\int \exp(\phi)dx\\
%        &\geq L_n(\phi).
%    \end{align*}
    $\phi$ is bounded on $[X_1,X_N]$, due to the continuity, i.e. $|L_n(\phi)|\leq C_{\phi}$ for some constant $C_{\phi}$, 
    %Therefore, $L_n(\phi)$ is a submartingale. 
    by Doob's submartingale convergence theorem, $L_n(\phi)\to L_{\infty}(\phi)$ almost surely for all $\phi\in {\cal C}$.
\end{proof}

Concerning the consistency of the sequence of optimizers, we need to investigate the uniform convergence of the sequence of objective functions.  Corresponding to an equicontinuity condition for deterministic functions, so we need to consider stochastic equicontinuity for the uniform convergence of random functions, which is studied by~\cite{Newey_1991}. We refer to~\cite{Davidson_2021} for the strong stochastic equicontinuity to further obtain the corresponding almost sure results.

\begin{definition}[Stochastic equicontinuity]
    Let $(\Theta,\rho)$ be a metric space and $(\Omega,\mathcal{F},P)$ a probability space and let $\{G_n(\theta,\omega),n\in\mathbb{N}\}$ be a sequence of stochastic functions: $G_n: \Theta\times\Omega\mapsto\mathbb{R}$, $\mathcal{F}/\mathcal{B}$-measurable for each $\theta\in\Theta$. The sequence is said to be stochastically equicontinuous (in probability) if for all $\epsilon>0$, there exists a $\delta>0$ such that 
    \begin{equation}
        \limsup_{n\to\infty}P\left(\sup_{\theta\in\Theta}\sup_{\|\theta-\theta'\|\leq\delta}|G_n(\theta)-G_n(\theta')|\geq\epsilon\right)< \epsilon.
    \end{equation}
    It is said to be strongly stochastically equicontinuous if for all $\epsilon>0$, there exists a $\delta>0$ such that
    \begin{equation}
        P\left(\limsup_{n\to\infty}\sup_{\theta\in\Theta}\sup_{\|\theta-\theta'\|\leq\delta}|G_n(\theta)-G_n(\theta')|\geq\epsilon\right)=0.
    \end{equation}
\end{definition}

Theorem 22.8 in Davidson (2021)~\cite{Davidson_2021} provides the relationship between strongly stochastic equicontinuity and almost sure uniform convergence:

\begin{theorem}
\label{thm:Davidson}
    Let $\{G_n(\theta),n\in\mathbb{N}\}$ be a sequence of stochastic real-valued functions on a totally bounded metric space $(\Theta,\rho)$. Then
    \begin{equation}
        \sup_{\theta\in\Theta}|G_n(\theta)| \overset{\text{a.s.}}{\to}0
    \end{equation}
    if and only if
    \begin{enumerate}[label=(\alph*)]
        \item $G_n(\theta) \overset{\text{a.s.}}{\to}0$ for each $\theta\in\Theta_0$ where $\Theta_0$ is the dense subset of $\Theta$; and
        \item $\{G_n\}$ is strongly stochastically equicontinuous.
    \end{enumerate}
\end{theorem}

\noindent
We now use Theorem \ref{thm:Davidson} to prove the following: 

\begin{theorem}[Almost sure convergence of the log-density]
\label{thm:Almost sure convergence of the log-density}
    Suppose $\Phi \subset \mathcal{C}$ is a compact subset 
    with respect to the sup metric, 
    and $L_n(\phi)$ converges to  $L_{\infty}(\phi)$ almost surely for all $\phi$, which is uniquely maximized at $\widehat\phi_\infty$. Set $\widehat{\phi}_n={\arg\max}_{\phi\in\Phi}L_n(\phi)$ for $n>N$. Then {$L_{\infty}$ is concave and} $\widehat{\phi}_n$ converges to $\widehat{\phi}_{\infty}$ almost surely.
\end{theorem}
\begin{proof}
    Define the metric for any $\phi, \phi'\in\Phi$:
    $$
    \|\phi-\phi'\| = \sup_{x\in[X_1,X_N]}|\phi(x)-\phi'(x)|.
    $$
    {Due to the compactness of $\Phi$ and the support of $\exp(\phi)$ and the convexity of the exponential function, ${L_n(\phi)}$ for $ n>N$ is an equi-Lipschitz-continuous family: for any $\phi,\phi'\in\Phi$, $n>N$, {there exists a $\Phi$-dependent constant $\beta_{\Phi}$, such that}
    %{\color{red} $$\cancel{
    %|L_n(\phi)-L_n(\phi')| \leq \int |\phi(x)-\phi'(x)|\,dF_n(x)\leq \|\phi-\phi'\|.
    %}$$ }
    {$$
    |L_n(\phi)-L_n(\phi')| \leq \beta_{\Phi}\|\phi-\phi'\|.
    $$}
    By the almost sure convergence {and the concavity} of $L_n(\phi)$,
    \begin{equation}
        |L_{\infty}(\phi)-L_{\infty}(\phi')| \leq {\beta_{\Phi}}\|\phi-\phi'\|,
    \end{equation}
    {and $L_{\infty}$ is also concave} almost surely for all $\phi,\phi'\in\Phi$. Therefore, for any $n$ and any $\phi,\phi'\in\Phi$,
    \begin{align*}
        |(L_{n}(\phi)-L_{\infty}(\phi))-(L_{n}(\phi')-L_{\infty}(\phi'))|
        \leq |L_{n}(\phi)-L_{n}(\phi')|+|L_{\infty}(\phi)-L_{\infty}(\phi')|
        \leq 2 {\beta_{\Phi}} \|\phi-\phi'\|
    \end{align*}
    almost surely. So $\{L_{n}(\phi)-L_{\infty}(\phi)\}$ is strongly stochastically equicontinuous.}

    Note that compactness is stronger than total boundedness. By Theorem~\ref{thm:Davidson}, $L_{n}(\phi)-L_{\infty}(\phi)$ uniformly converges to 0 almost surely. Due to the compactness of $\Phi$ and $[X_1,X_N]$ and the strict concavity of $L_n$, it is that $\widehat{\phi}_n$ converges to $\widehat{\phi}_{\infty}$ almost surely.
\end{proof}

\noindent
Note that $\{L_n\}$ is a family of strict concave functions. Recalling the $L_n$ in Equation~\ref{equ:npmle_loss}, for any $\|\phi-\phi'\|>0$, there exists an $x_0$ such that $\phi(x_0)\neq\phi'(x_0)$ and without loss of generality, we assume $\phi(x_0)-\phi'(x_0)>0$. Due to the continuity of $\phi$ and $\phi'$, there exists a $\delta>0$, such that for any $x\in[x_o-\delta,x_0+\delta]$, $\phi(x)-\phi'(x)>0$. So $\exp(\alpha\phi(x)+(1-\alpha)\phi'(x))<\alpha\exp(\phi(x))+(1-\alpha)\exp(\phi'(x))$ for any $x\in[x_o-\delta,x_0+\delta]$. Due to the continuity and convexity of the exponential function, 
$$\int_{X_1}^{X_N}\exp\left(\alpha\phi(x)+(1-\alpha)\phi'(x)\right)\,dx<\alpha\int_{X_1}^{X_N}\exp(\phi(x))dx+(1-\alpha)\int_{X_1}^{X_N}\exp(\phi'(x))dx$$
for any $\alpha\in[0,1]$. Since $\int\phi\, d F_n$ is linear in $\phi$, so $L_n$ is strictly concave.

In Theorem~\ref{thm:Almost sure convergence of the log-density}, we assume $L_{\infty}$ uniquely maximized at $\widehat{\phi}_{\infty}$. Even if the optimal solution is not unique for $L_{\infty}$, due to the strict concavity of $\{L_n\}$, $\widehat{\phi}_n$ is unique for each $n$ and will converge to an element within the optimal solution set of $L_{\infty}$.

These results are related to~\cite{Newey_1991} using stochastic equicontinuity and the Arzel\`a–Ascoli Theorem. %But we have additional concavity and equi-Lipschitz-continuity in our settings.
The results in \cite{Newey_1991} apply to our setting due to the stochastic equicontinuity of the $(L_n(\phi))_{n>N}$.

\begin{theorem}[Asymptotic exchangeability]
\label{thm:Asymptotic exchangeability}
    Under the same assumptions of Theorem~\ref{thm:Almost sure convergence of the log-density}, $X_{N+1:\infty}$ sampled from the bootstrap Algorithm~\ref{algo:1} is asymptotic exchangeable.
\end{theorem}
\begin{proof}
    Due to the almost sure convergence of $\widehat{\phi}_n$, by continuous mapping theorem, $\widehat{f}_n = \exp(\widehat{\phi}_n)$ also converges almost surely. For any bounded and continuous function $g$, due to the compactness of $\Phi$ and Lebesgue dominated convergence theorem,
    $$
    \left|\int g \widehat{f}_n   - \int g \widehat{f}_{\infty}  \right| \leq \max|g| \int|\widehat{f}_n- \widehat{f}_{\infty}|  \to 0 
    $$
    as $n\to\infty$. So the corresponding random distributions converge weakly almost surely. The tightness is obvious because they are all supported on $[X_1,X_N]$. By Lemma 8.2(b) in~\cite{Aldous_1985}, $X_{N+1:\infty}$ is asymptotically exchangeable.
\end{proof}

\section{Illustrations}
\label{sec:Illustration}
We now present some illustrations for the martingale posterior as described in Algorithm \ref{algo:1}. It is worth noting that one important way to accelerate the algorithm is to find a faster way to calculate the NPMLEs. For example, we use the LC algorithm in \cite{Anderson_2014} or the constrained Newton method for log-concave density estimation (CNMLCD) algorithm; see \cite{Liu_2018}.

Before proceeding, we show how the NPMLE provides a good estimator of a log concave density. 
\begin{figure}[ht!]
    \centering
    \includegraphics[width=0.8\textwidth]{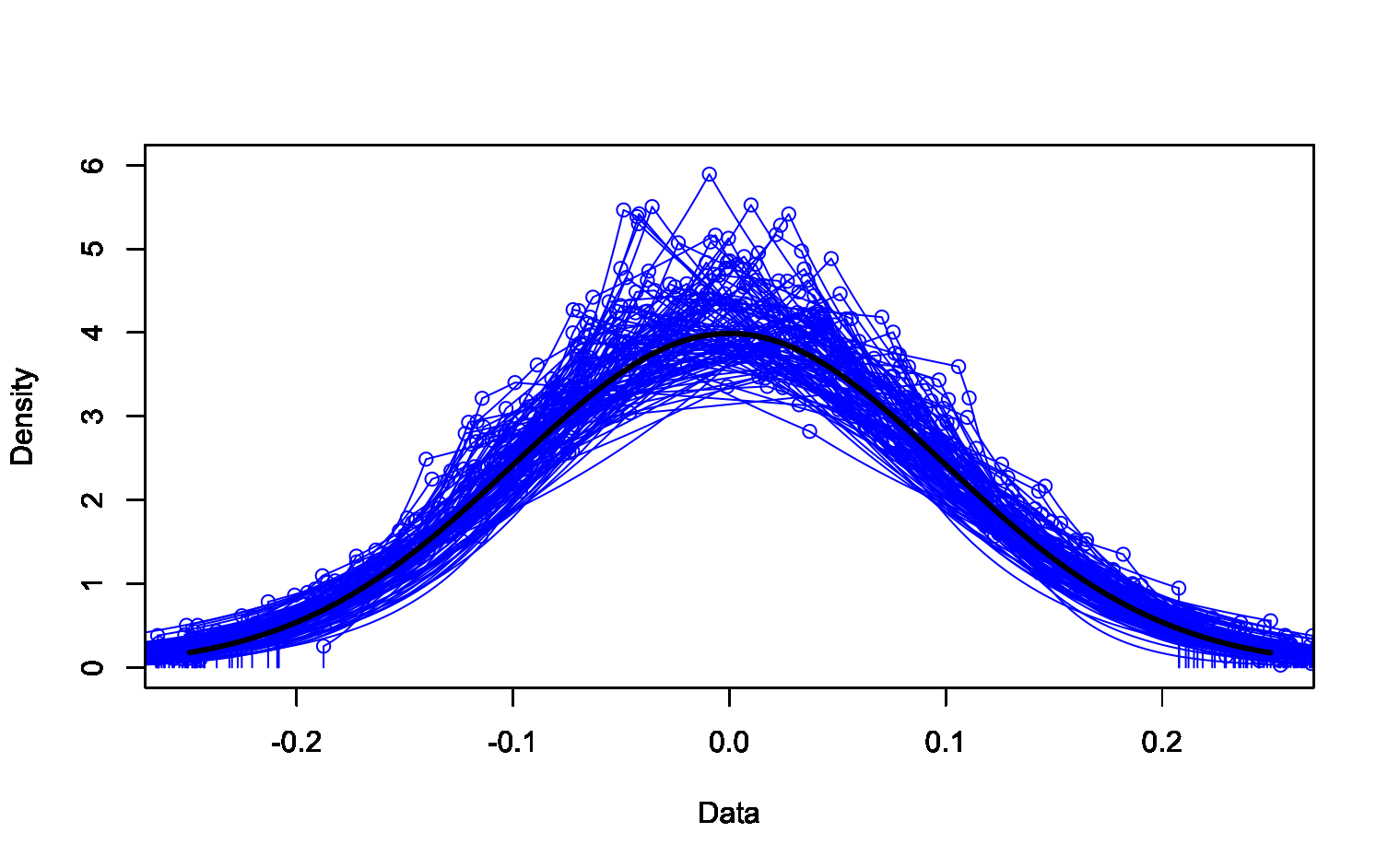}
    \caption{NPMLE estimators}
    \label{pl:NPMLE_avg}
\end{figure}

While some may appear biased and even spikey, we see from Fig. \ref{pl:NPMLE_avg}
that on average the density estimators are good on representations of the true density.  
%From the experiments, we can see the mean of our posteriors could be higher than the true ones, especially in some high-frequency areas. One reason is the biasedness of our method. Another could be the influence of the tails. Using NPMLE will discard the probability of $\{x\notin[X_1,X_n]\}$. For some heavy-tail probability, we can observe this problem in the experiments. Moreover, some biasedness in the illustration may be resulted from NPMLE itself. On average NPMLEs will estimate the sample distribution density well, but due to the randomness, we may find some biasedness in some specific estimates. However, our bootstrap method to the NPMLE still works well. See Figure~\ref, 
For the figure we took $n=200$ iid samples from $\mathcal{N}(0,0.1^2)$ 100 times and then plotted the 100 NPMLEs with the true normal density function. 

In each of the examples, $M$ was chosen conservatively, i.e. rounded up to the nearest 1000 for Section~\ref{subsec:Simulated_Data} and~\ref{subsec:Real-world_Data}, and 100 for Section~\ref{subsec:Comparison}, so that the difference between $\widehat{F}_M$ and $\widehat{F}_{M+1}$ is suitably small, to the order of $10^{-3}$ or $10^{-2}$.  Hence, in Sections 3.1 and 3.2 we take $M=1000$ and in Section 3.3 we use $M=100$.   

\subsection{Simulated data}
\label{subsec:Simulated_Data}
First, we generate 70 and 1000 samples separately from the $\mathcal{N}(0,1)$, $\mathrm{Exp}(5)$ and $\mathrm{Laplace}(0,1)$ distributions, implying 6 experiments in total. For each of these experiments, we generate 50 posterior samples, and with each sample obtained by 1000 iterations, that is $M=1000$. The sampling can be performed using parallel computing.

\begin{figure}[!ht]
\centering
\subfigure[]{
\label{pl:normal_70}
\includegraphics[width=0.3\textwidth]{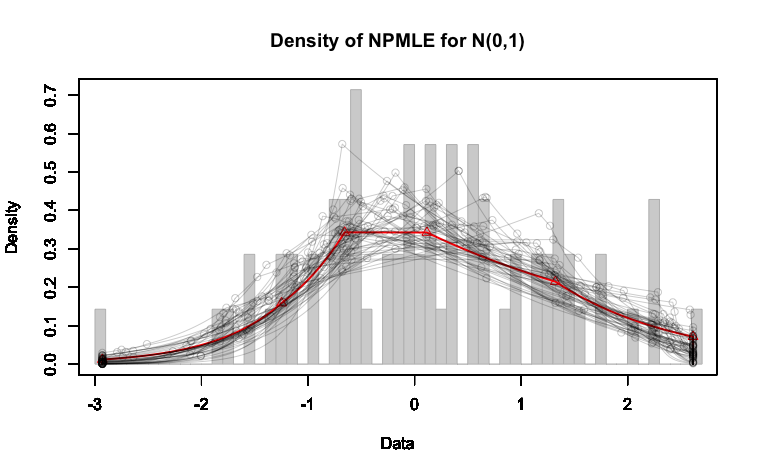}}
\subfigure[]{
\label{pl:exp5_70}
\includegraphics[width=0.3\textwidth]{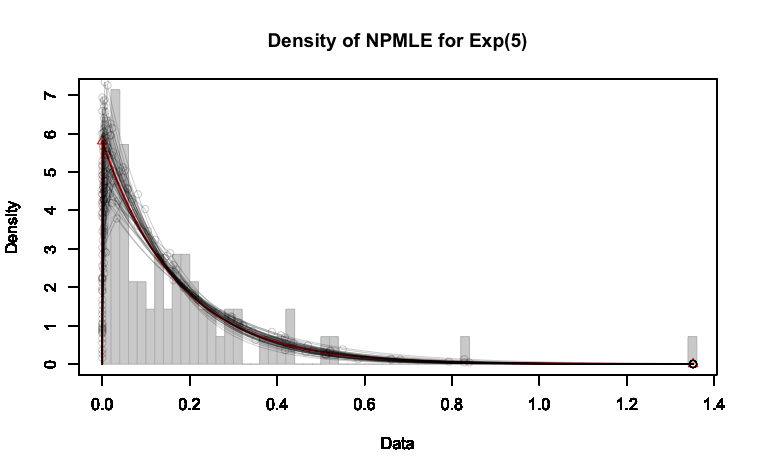}}
\subfigure[]{
\label{pl:laplacian_70}
\includegraphics[width=0.3\textwidth]{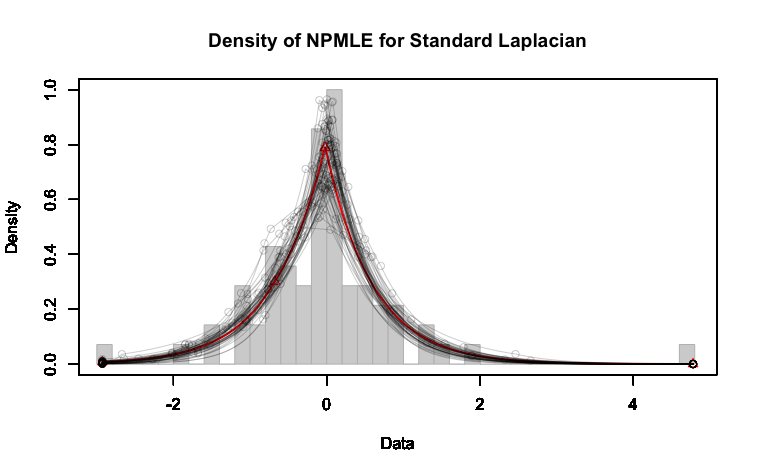}}
\subfigure[]{
\label{pl:normal_log_70}
\includegraphics[width=0.3\textwidth]{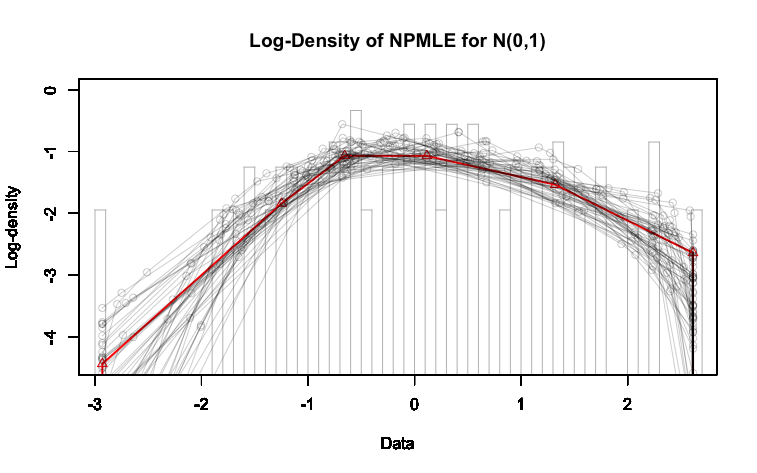}}
\subfigure[]{
\label{pl:exp5_log_70}
\includegraphics[width=0.3\textwidth]{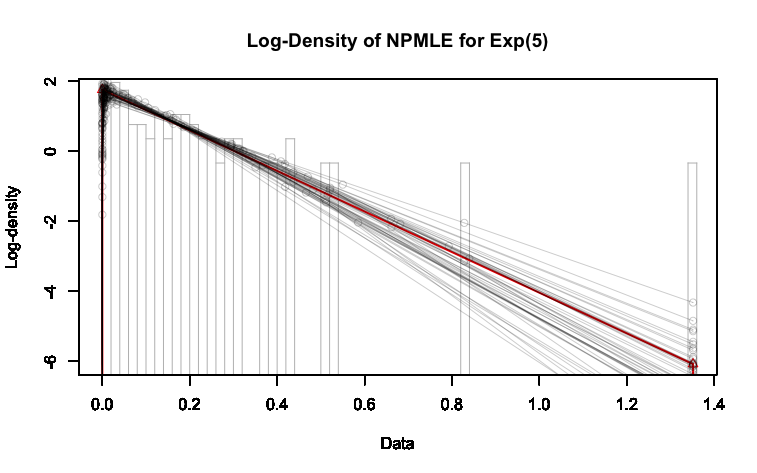}}
\subfigure[]{
\label{pl:laplacian_log_70}
\includegraphics[width=0.3\textwidth]{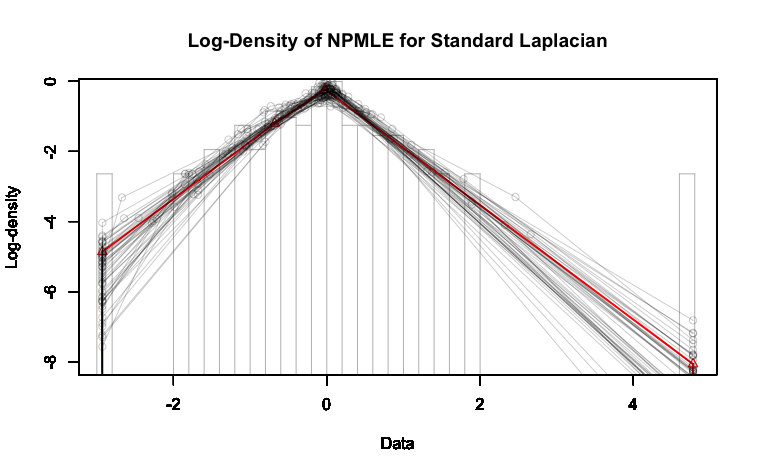}}
\caption{Functional Bootstrapping for some log-concave distributions with 70 original samples}
\label{pl:normal_exp_laplacian_70}
\end{figure}

\begin{figure}[!ht]
\centering
\subfigure[]{
\label{pl:normal}
\includegraphics[width=0.3\textwidth]{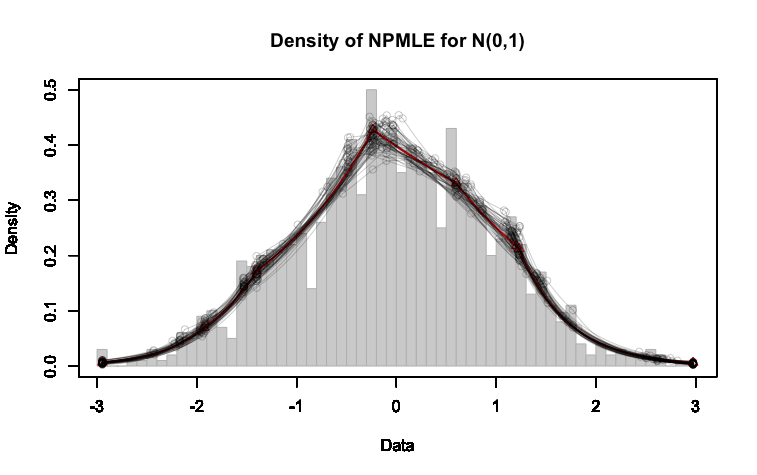}}
\subfigure[]{
\label{pl:exp5}
\includegraphics[width=0.3\textwidth]{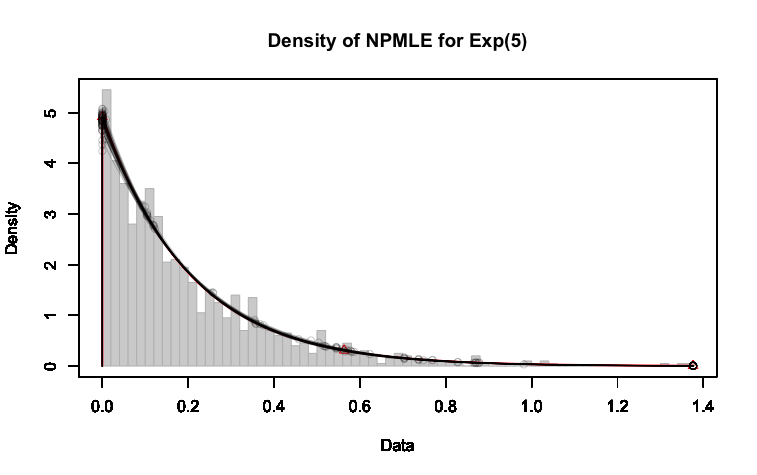}}
\subfigure[]{
\label{pl:laplacian}
\includegraphics[width=0.3\textwidth]{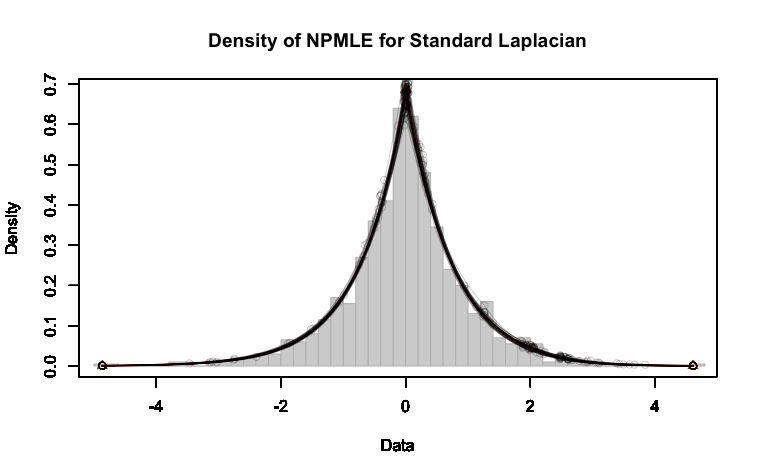}}
\subfigure[]{
\label{pl:normal_log}
\includegraphics[width=0.3\textwidth]{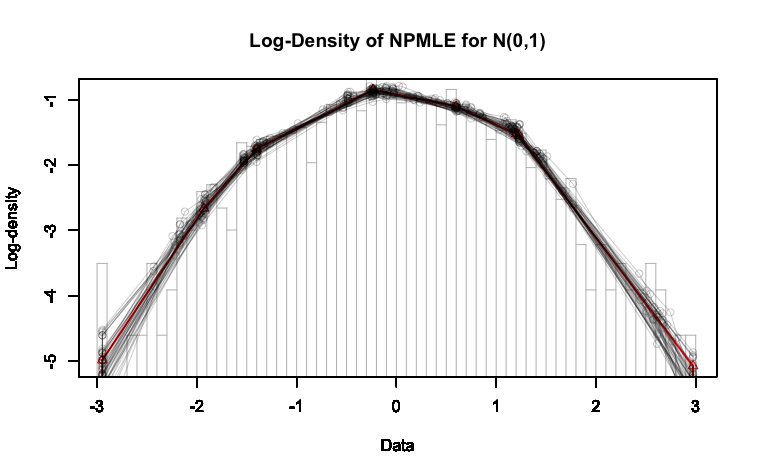}}
\subfigure[]{
\label{pl:exp5_log}
\includegraphics[width=0.3\textwidth]{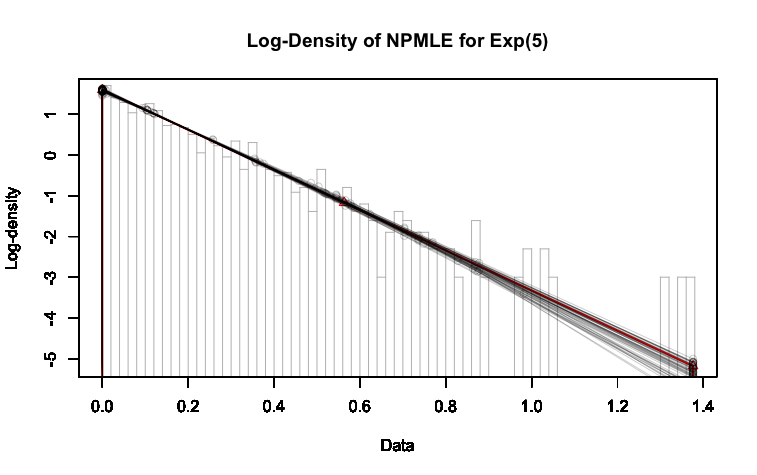}}
\subfigure[]{
\label{pl:laplacian_log}
\includegraphics[width=0.3\textwidth]{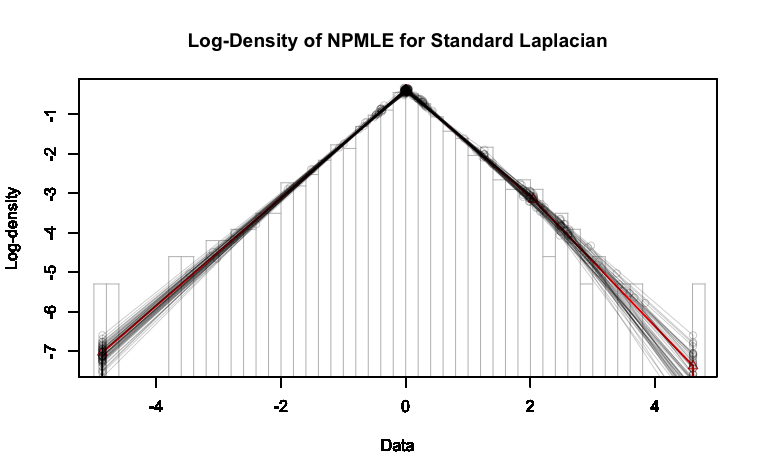}}
\caption{Functional Bootstrapping for some log-concave distributions with 1000 original samples}
\label{pl:normal_exp_laplacian_1000}
\end{figure}

In Figs \ref{pl:normal_exp_laplacian_70} and \ref{pl:normal_exp_laplacian_1000} we present the results for the 70 and 1000 samples, respectively. In each figure, the NPMLE densities are represented by red lines. The black lines are posterior samples from these NPMLEs. The red triangles or black circles are the non-differentiable points for the NPMLE density functions and their corresponding posterior samples.

The top rows represent the density functions for, from left to right, the normal, exponential and Laplace cases, while the bottom rows are for the corresponding log-concave functions.

As we can see, the representation of the martingale posterior samples are centered about the NPMLE and reflect the shape of the underlying concave density. Also clear is that the uncertainty obtained from the 70 samples is larger than that from the 1000 samples.

\subsection{Real data}
\label{subsec:Real-world_Data}
For the real-world data, we use three data sets from~\cite{Liu_2018} in the R package \texttt{cnmlcd}. The first one is the daily log-returns of Standard and Poor's 500 (S\&P 500) from March 1st 2011 to March 1st 2012, with the sample size of 252. The second is the daily log-returns of S\&P 500 from January 2nd 2014 to December 31st 2014, with the sample size of 252. The third is the daily log-volatilities of S\&P 500 from January 3rd 1995 to January 3rd 2014, with the sample size of 4783. 
%All these Standard and Poor's 500 data are from Yahoo Finance.

\begin{figure}[!ht]
\centering
\subfigure[]{
\label{pl:log_return_2011}
\includegraphics[width=0.3\textwidth]{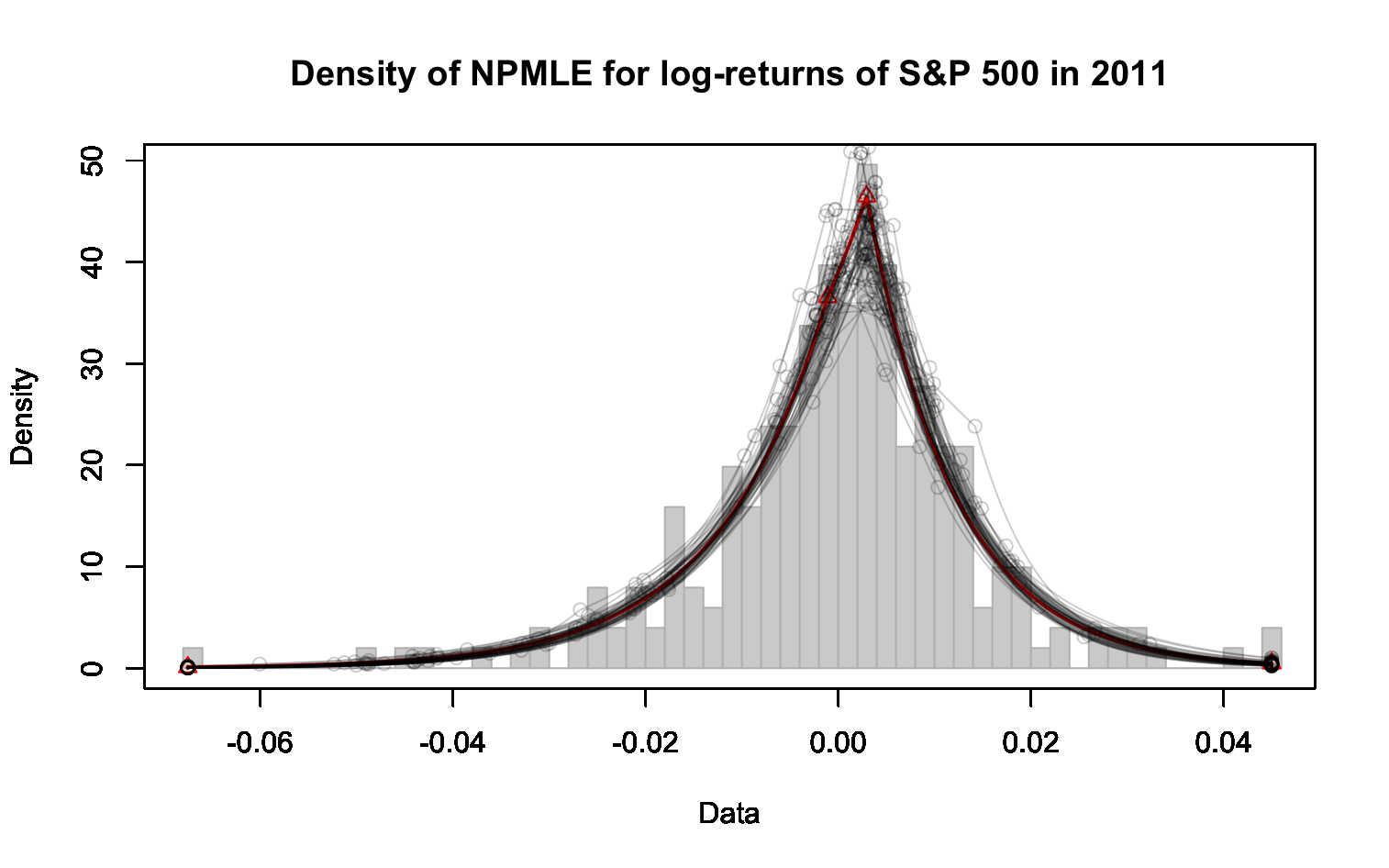}}
\subfigure[]{
\label{pl:log_return_2014}
\includegraphics[width=0.3\textwidth]{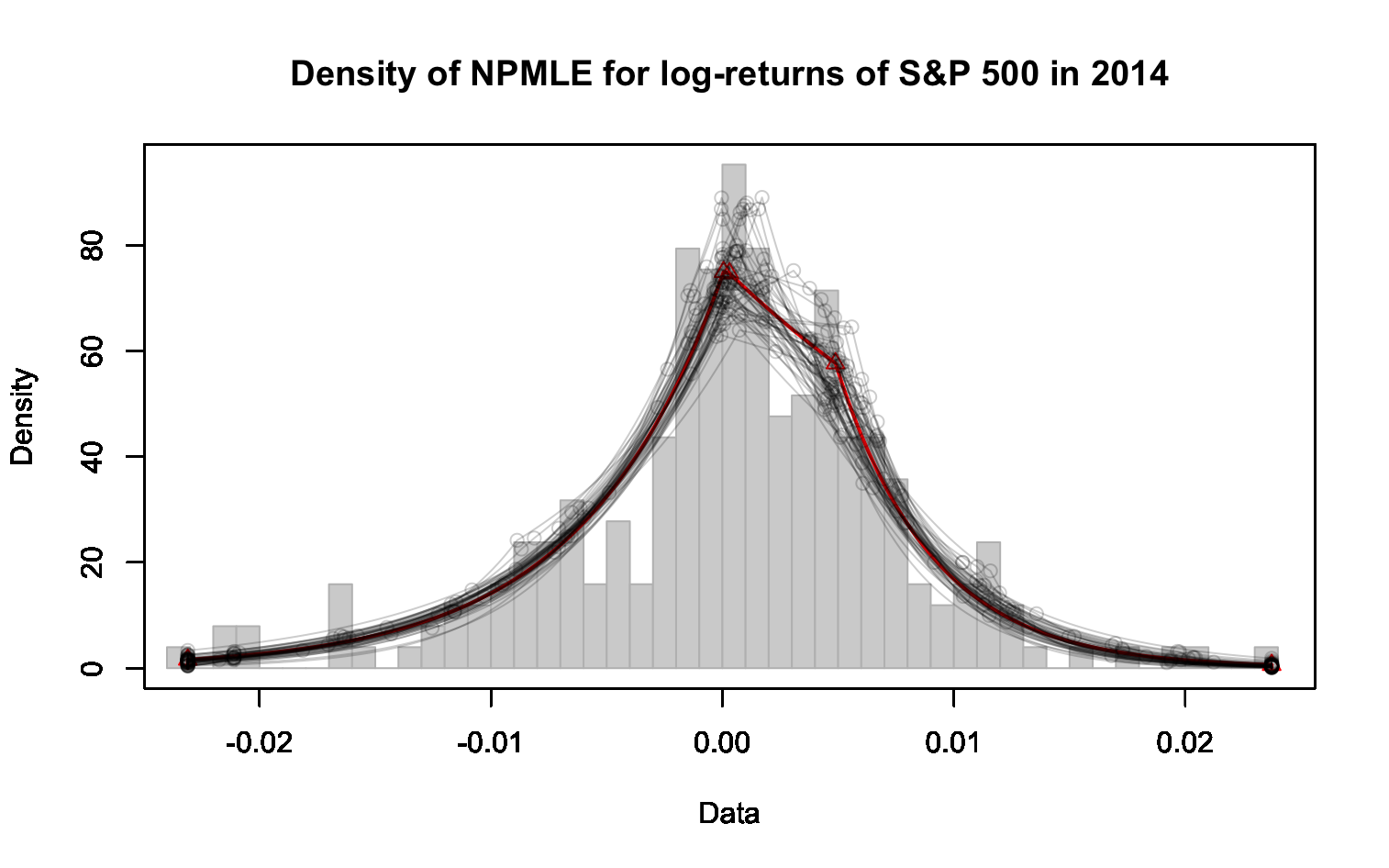}}
\subfigure[]{
\label{pl:log_volatility}
\includegraphics[width=0.3\textwidth]{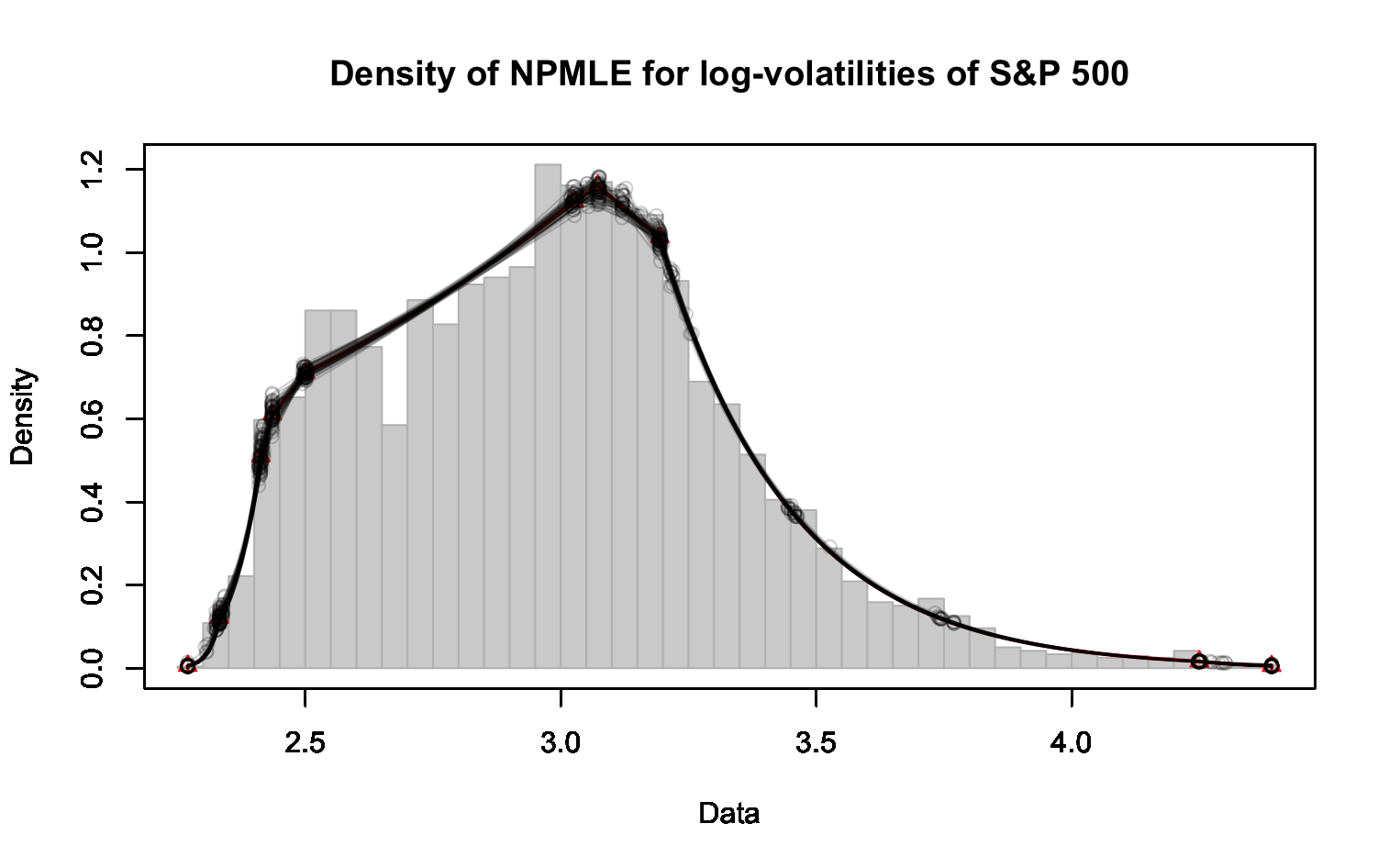}}
\subfigure[]{
\label{pl:log_return_2011_log_dens}
\includegraphics[width=0.3\textwidth]{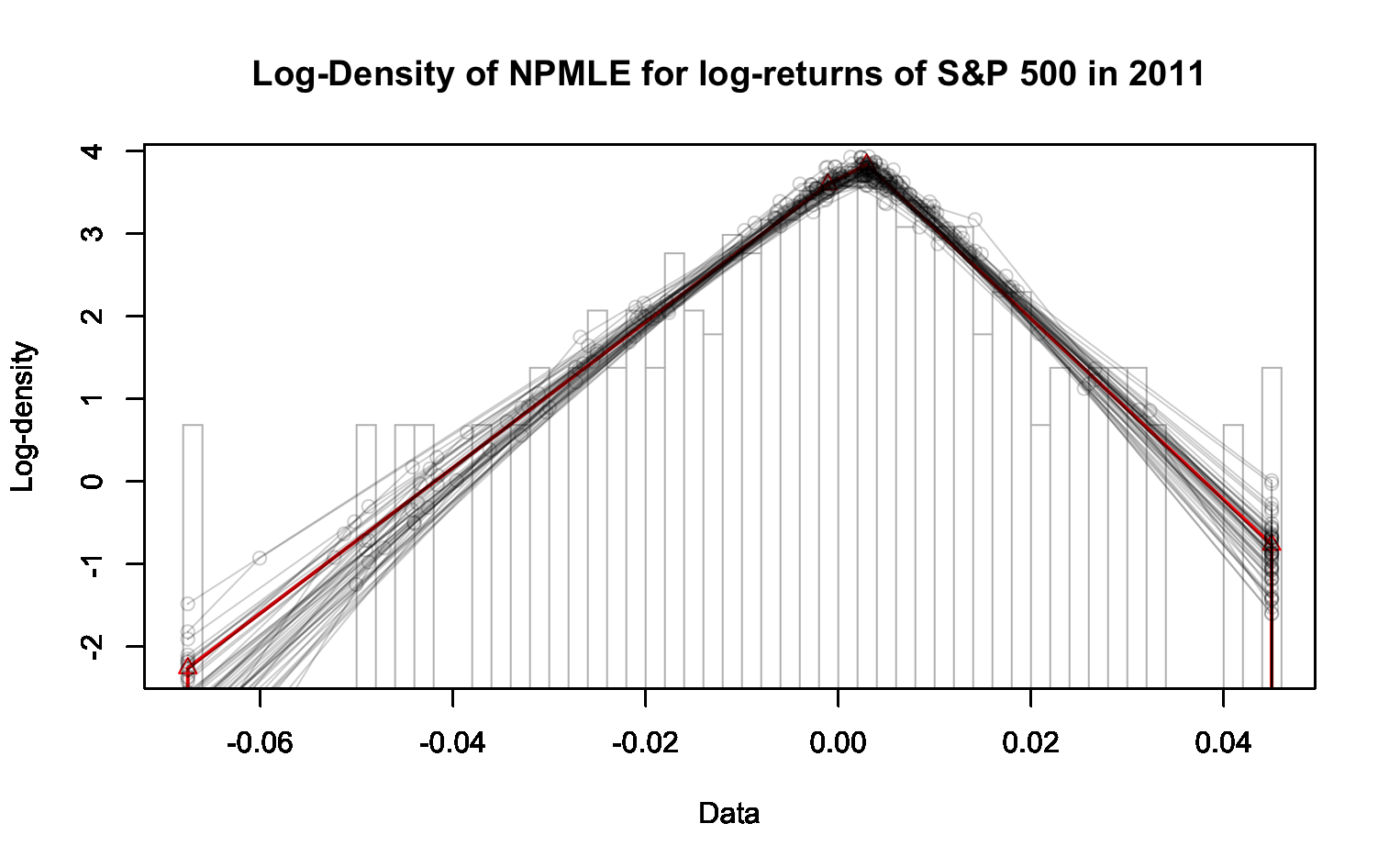}}
\subfigure[]{
\label{pl:log_return_2014_log_dens}
\includegraphics[width=0.3\textwidth]{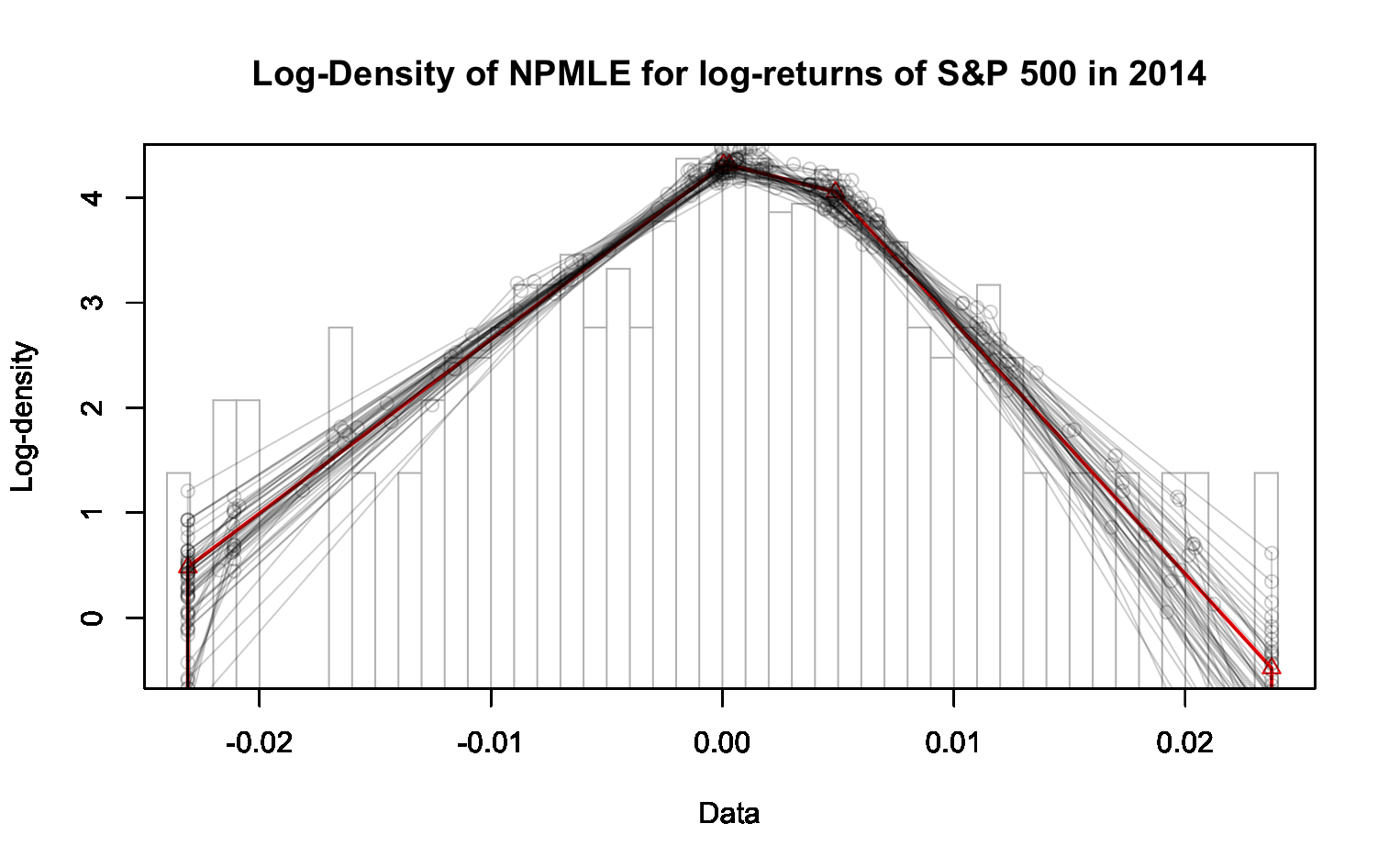}}
\subfigure[]{
\label{pl:log_volatility_log_dens}
\includegraphics[width=0.3\textwidth]{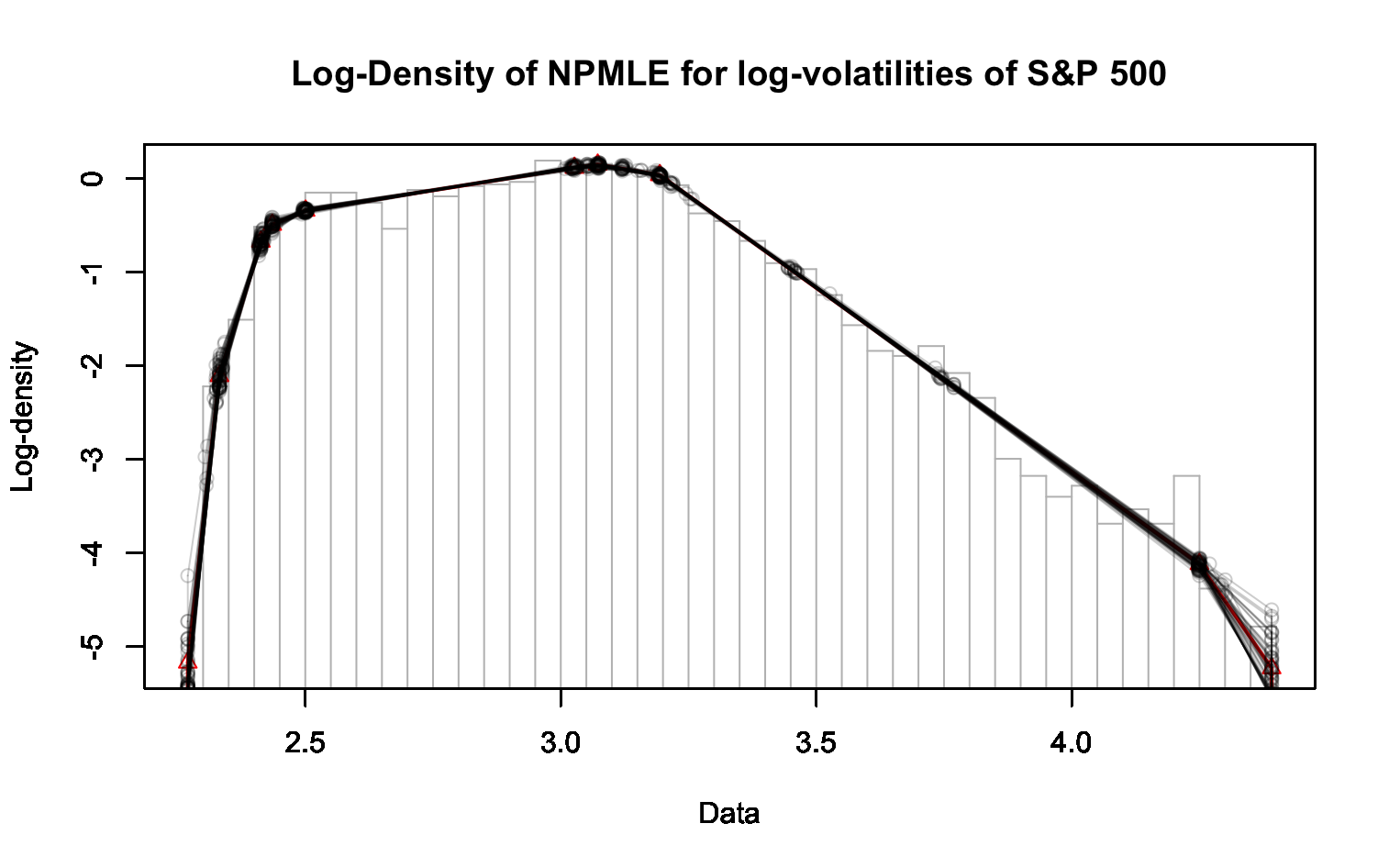}}
\caption{Real-world Data from S\&P 500. The first row is the posterior samples for the density functions of the log-returns in 2011, the log-returns in 2014 and log-volatilities. The second row shows the corresponding log-density functions of the posterior samples}
\label{pl:real-world_data}
\end{figure}

We first evaluate the NPMLEs, and then obtain 50 posterior samples for each of them, with each sample obtained by 1000 iterations of the martingale. We present the results in Figure~\ref{pl:real-world_data}. Again, due to the large sample sizes, the posterior uncertainty is quite tight about the NPMLE.

\begin{figure}[ht!]
\centering
\subfigure[]{
\label{pl:Comp_N0_S500_n50}
\includegraphics[width=0.3\textwidth]{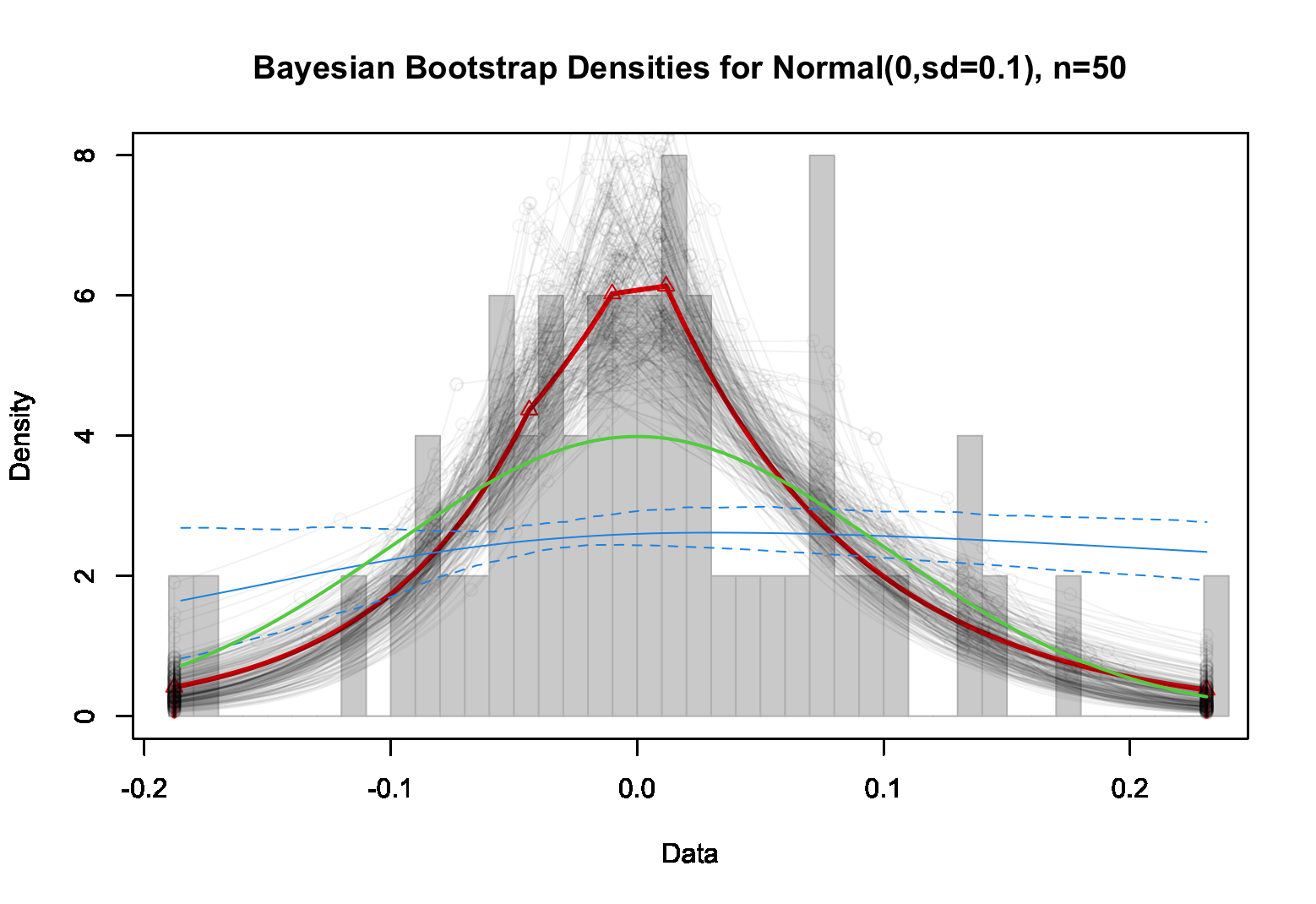}}
\subfigure[]{
\label{pl:Comp_N100_S500_n50}
\includegraphics[width=0.3\textwidth]{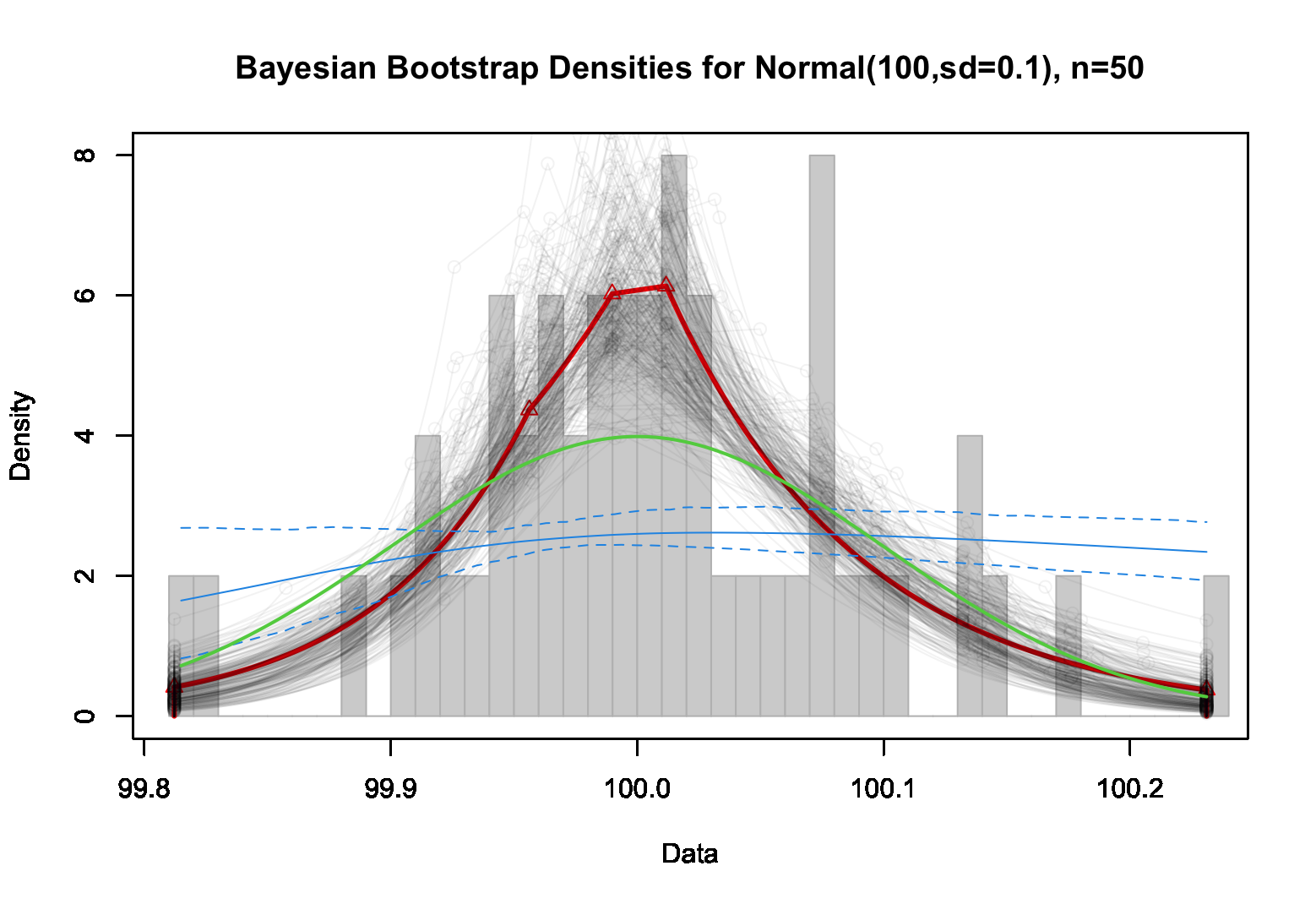}}
\subfigure[]{
\label{pl:Comp_G_S500_n50}
\includegraphics[width=0.3\textwidth]{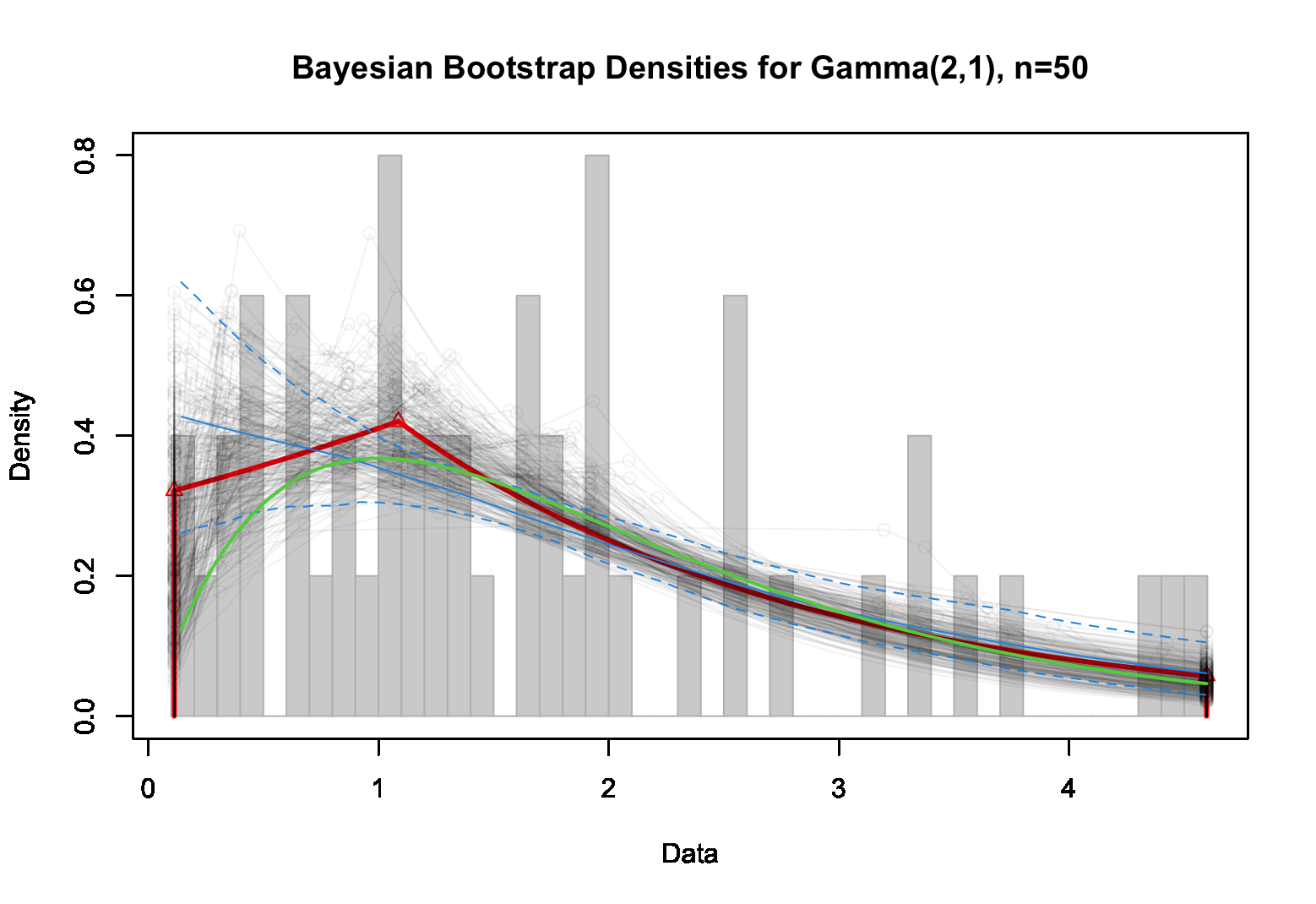}}
\subfigure[]{
\label{pl:Comp_N0_S500_n200}
\includegraphics[width=0.3\textwidth]{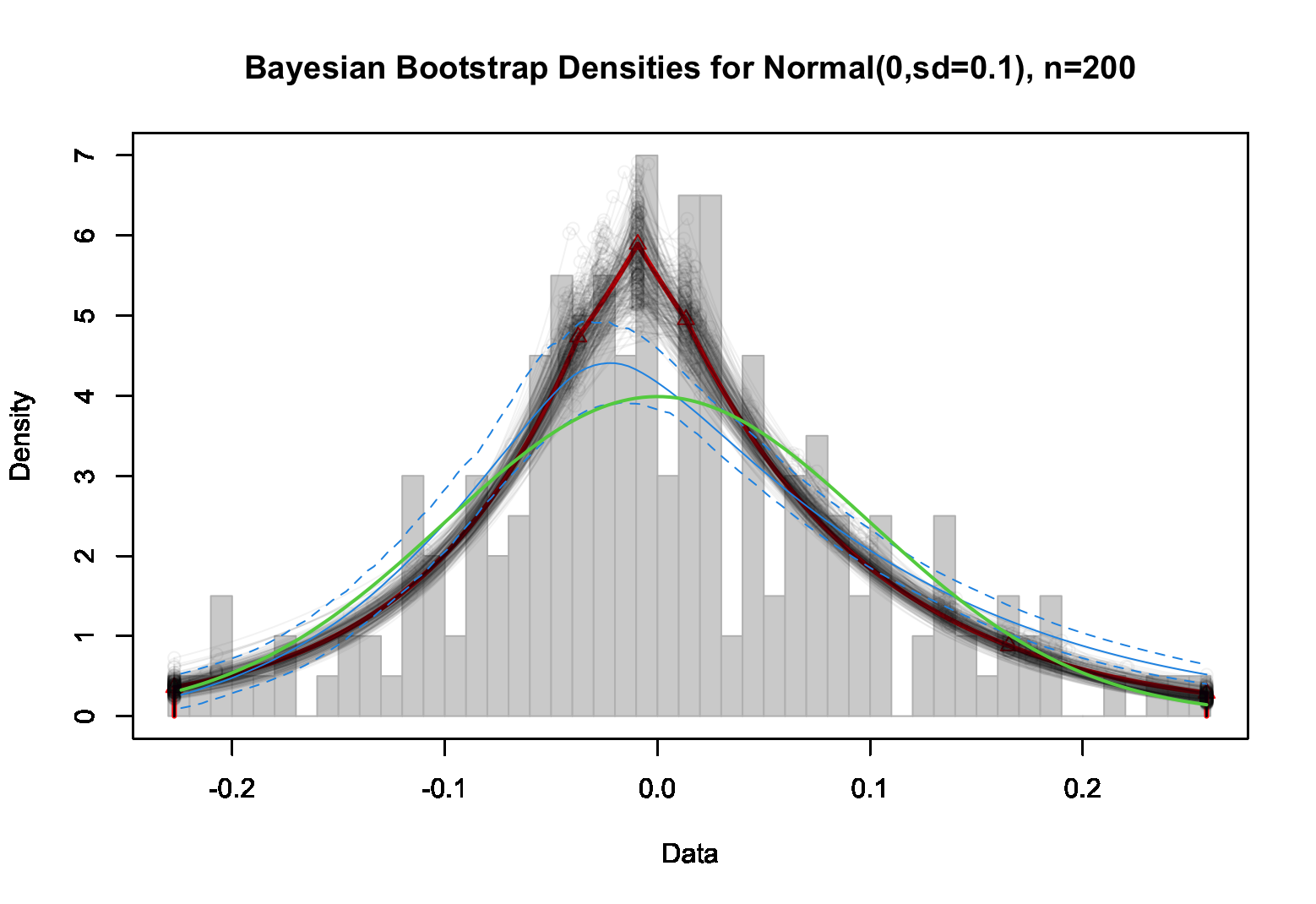}}
\subfigure[]{
\label{pl:Comp_N100_S500_n200}
\includegraphics[width=0.3\textwidth]{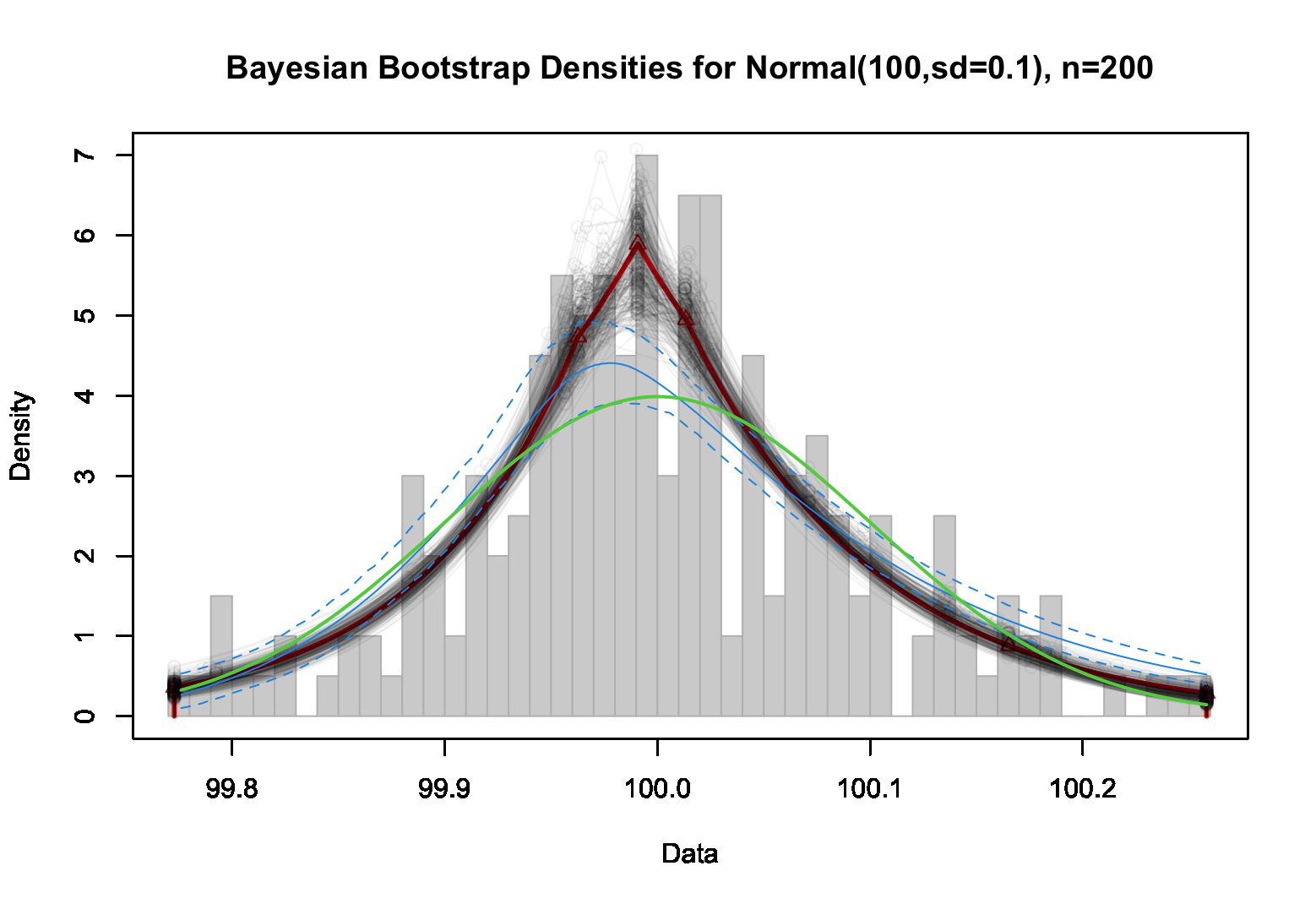}}
\subfigure[]{
\label{pl:Comp_G_S500_n200}
\includegraphics[width=0.3\textwidth]{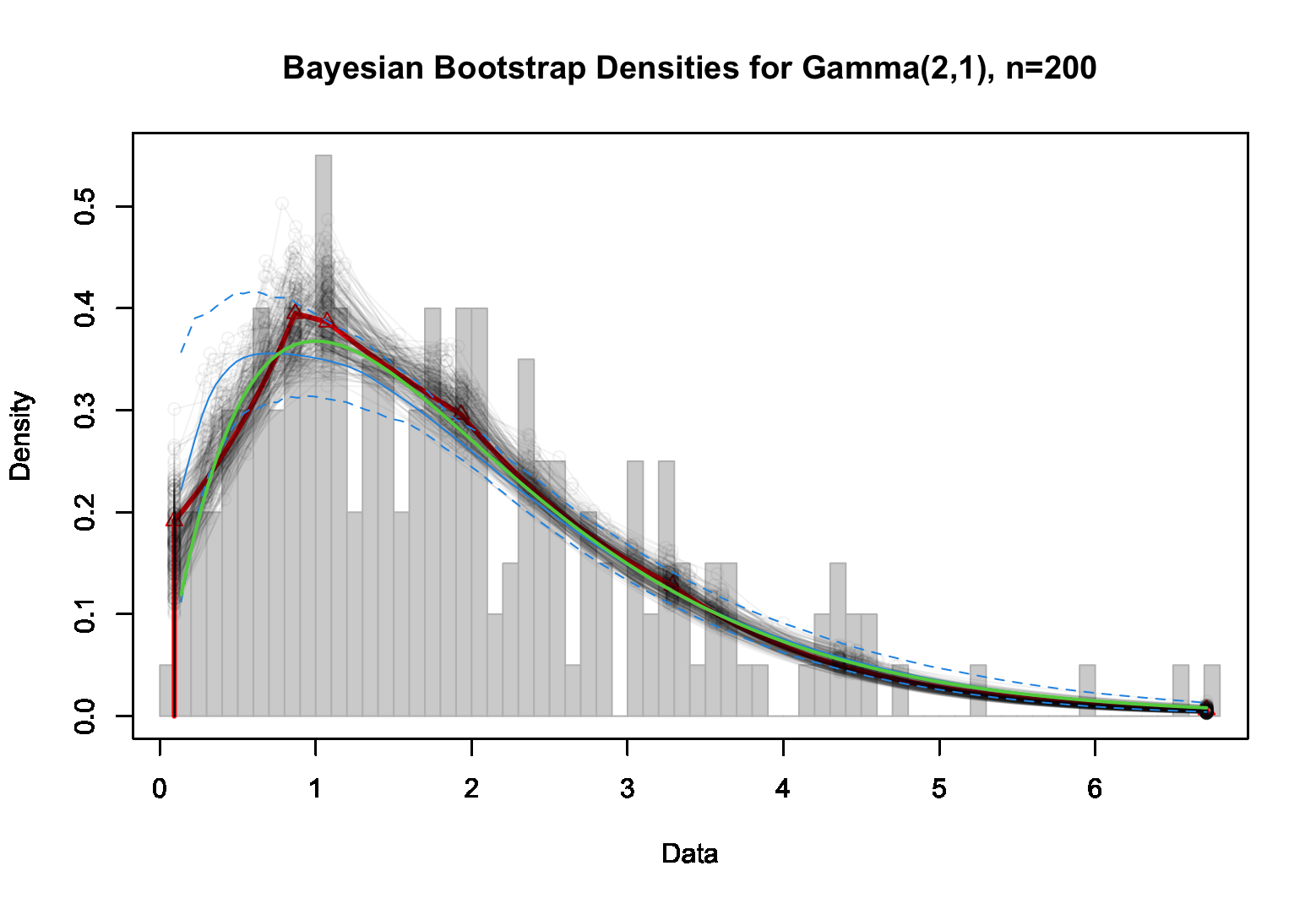}}
\subfigure[]{
\label{pl:Comp_N0_S500_n500}
\includegraphics[width=0.3\textwidth]{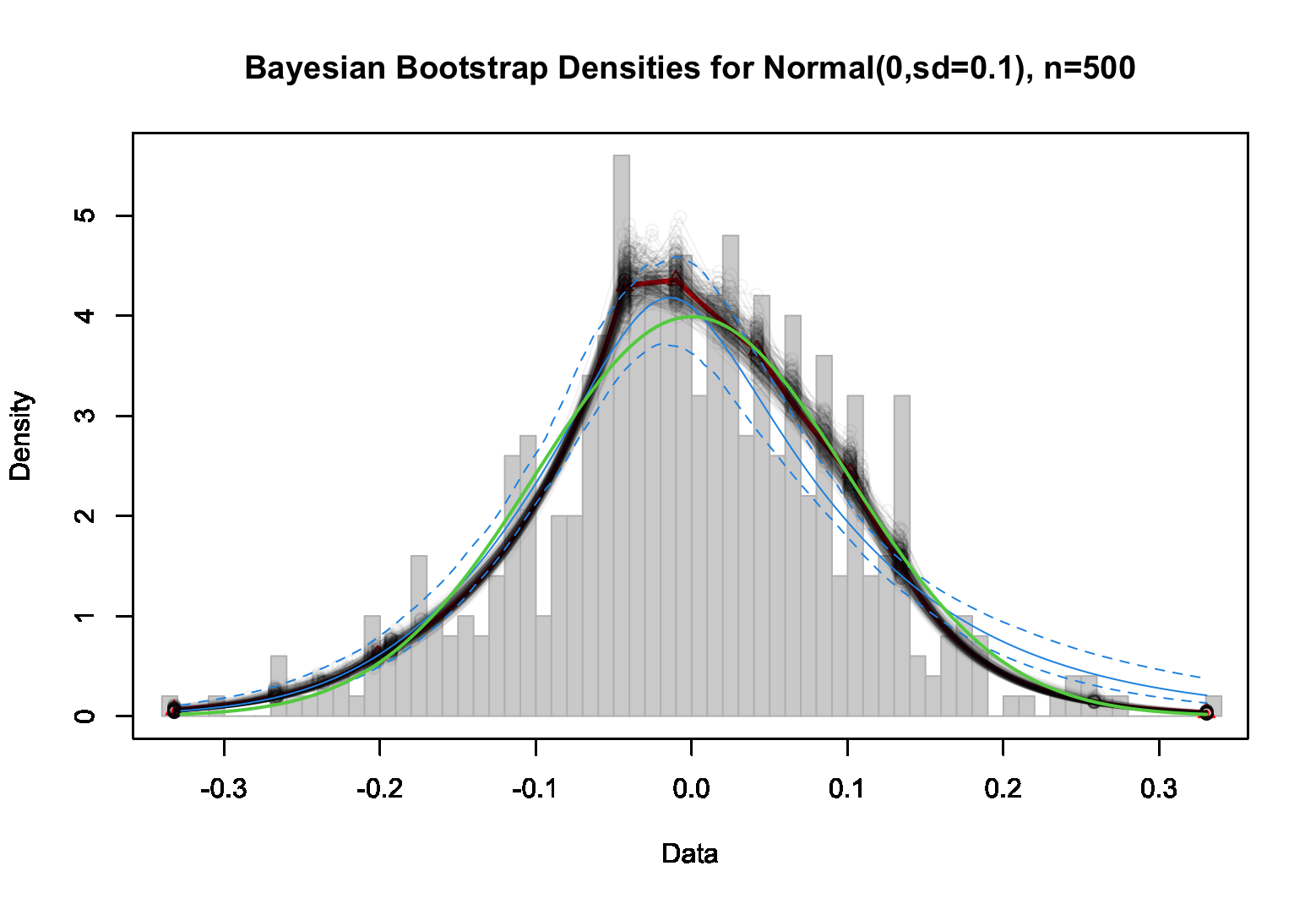}}
\subfigure[]{
\label{pl:Comp_N100_S500_n500}
\includegraphics[width=0.3\textwidth]{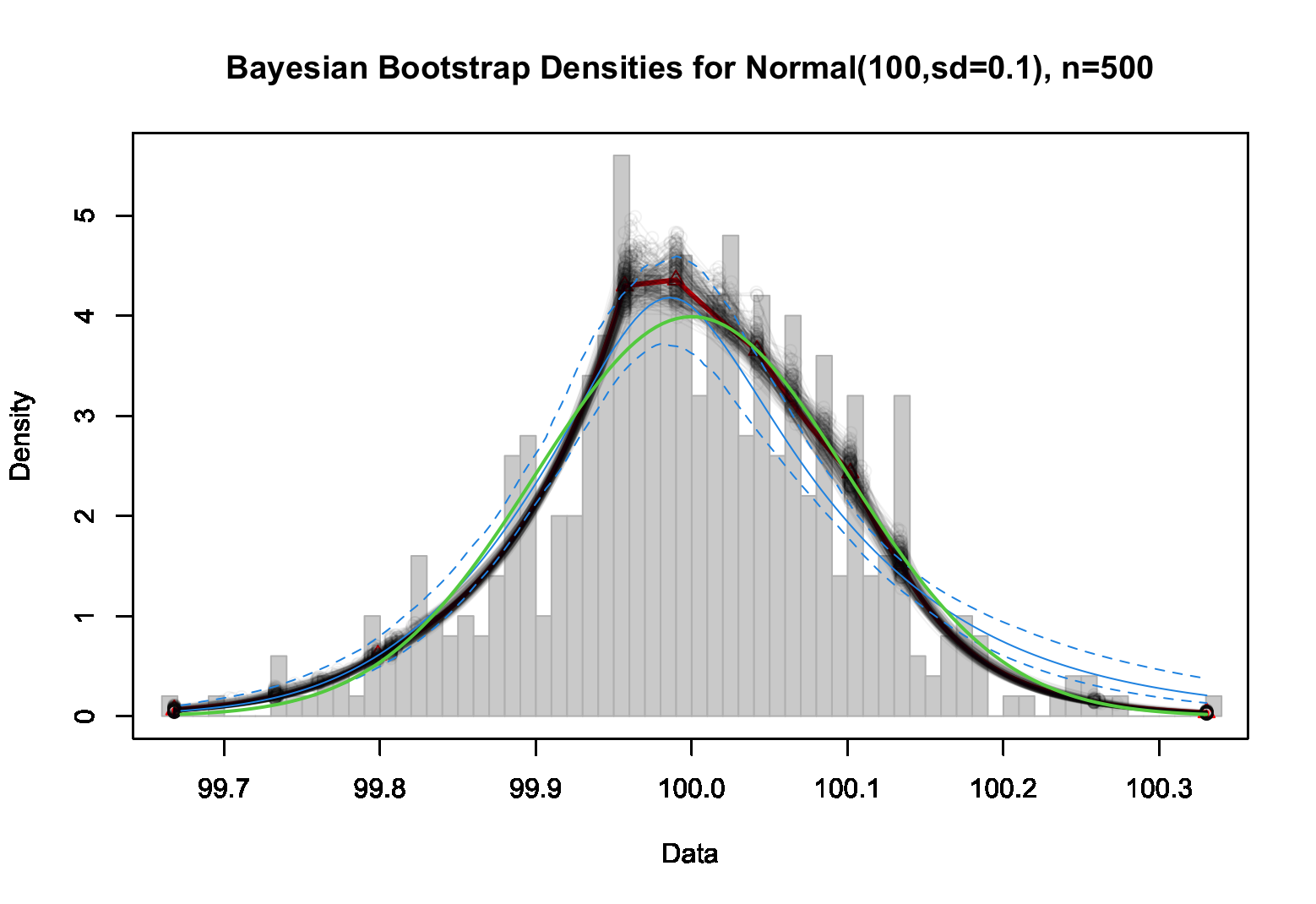}}
\subfigure[]{
\label{pl:Comp_G_S500_n500}
\includegraphics[width=0.3\textwidth]{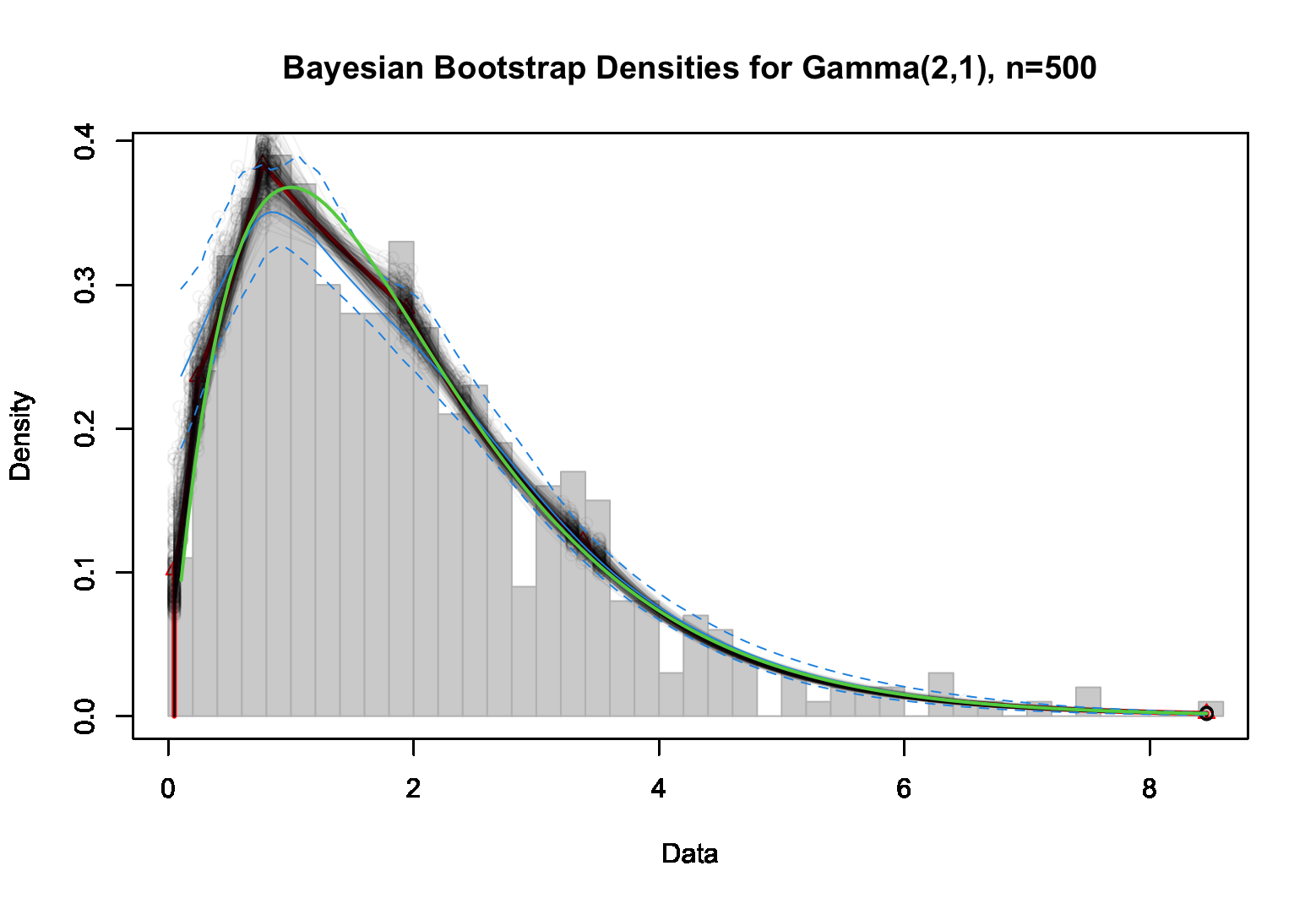}}
\subfigure[]{
\label{pl:Comp_N0_S500_n2500}
\includegraphics[width=0.3\textwidth]{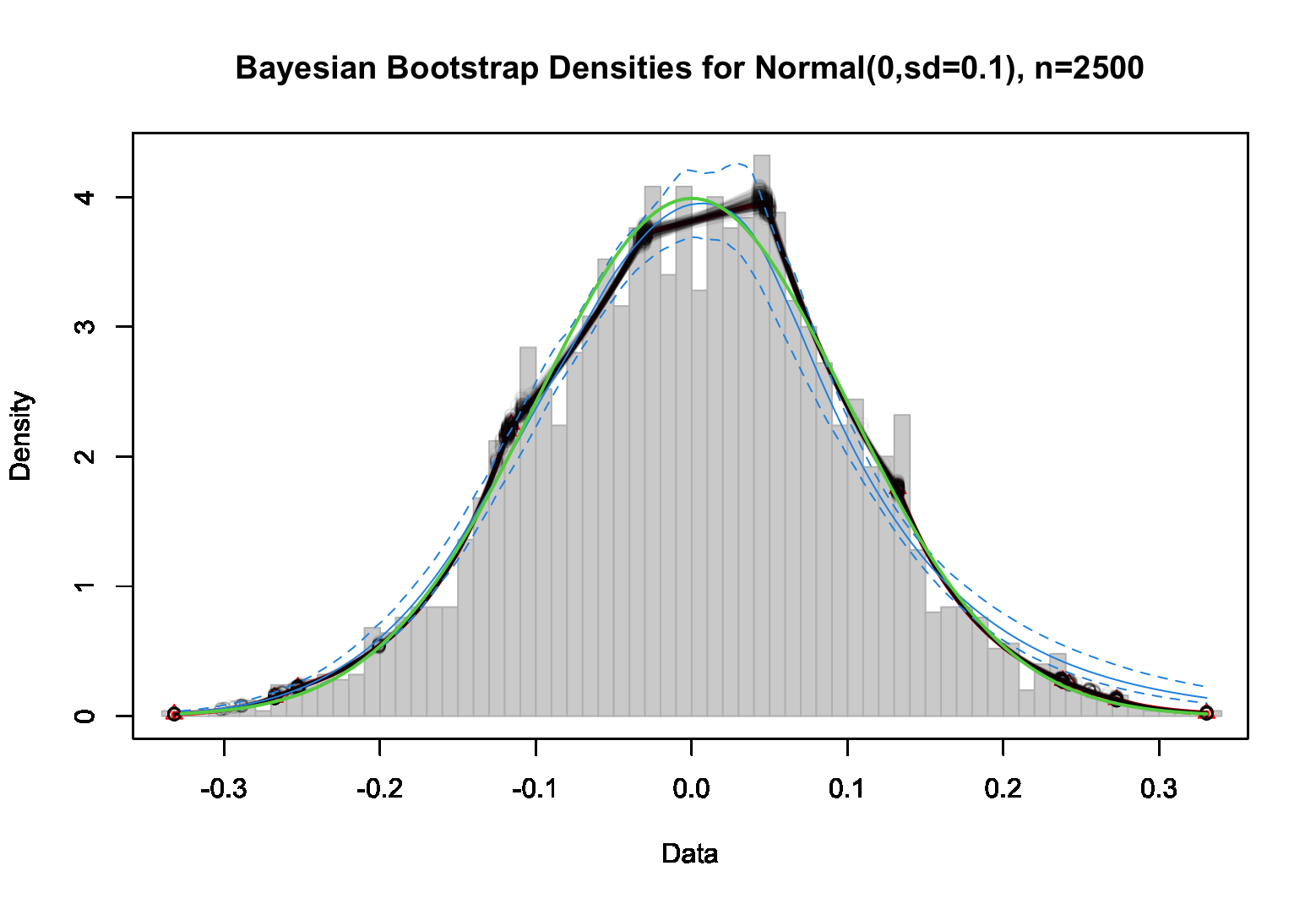}}
\subfigure[]{
\label{pl:Comp_N100_S500_n2500}
\includegraphics[width=0.3\textwidth]{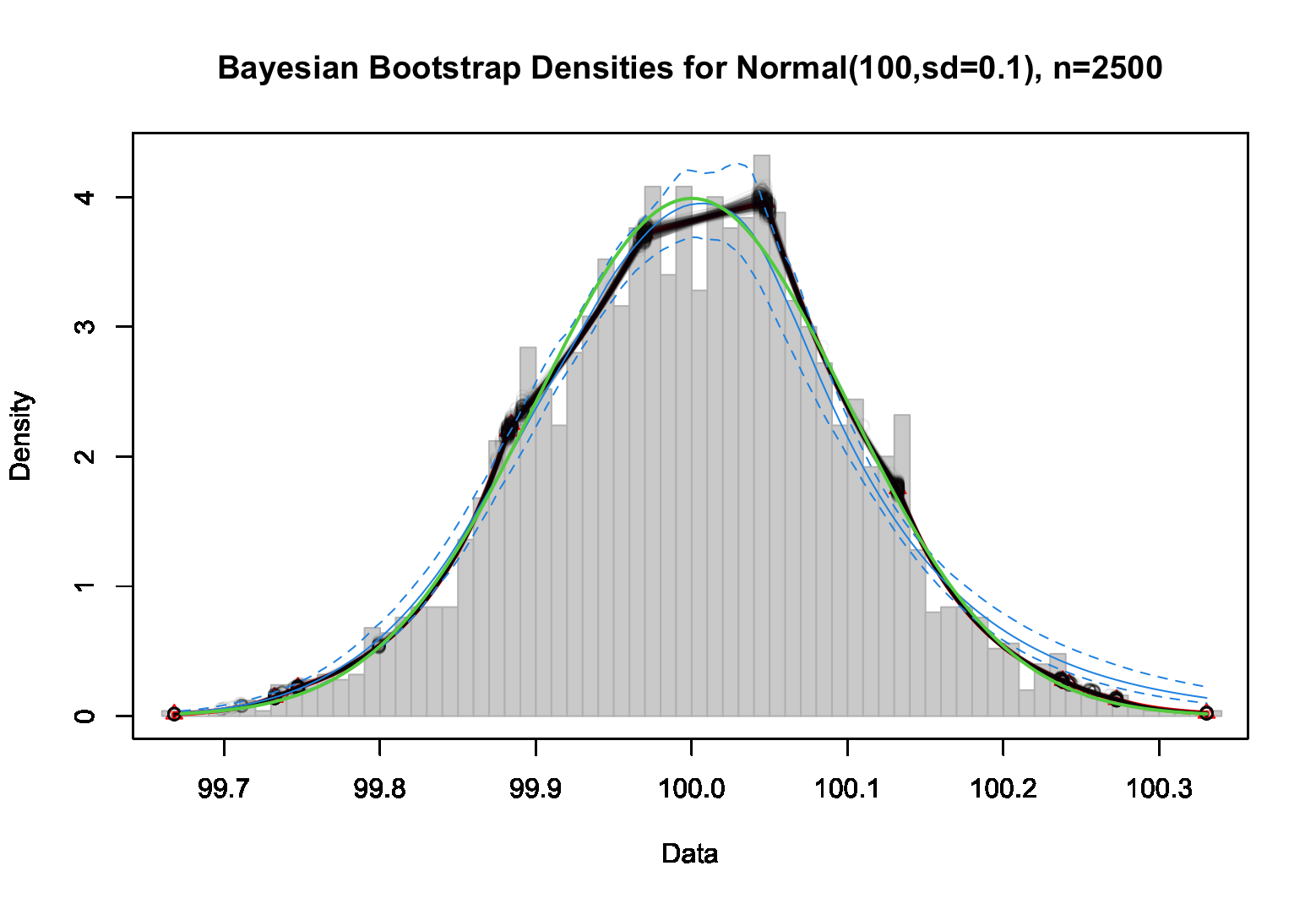}}
\subfigure[]{
\label{pl:Comp_G_S500_n2500}
\includegraphics[width=0.3\textwidth]{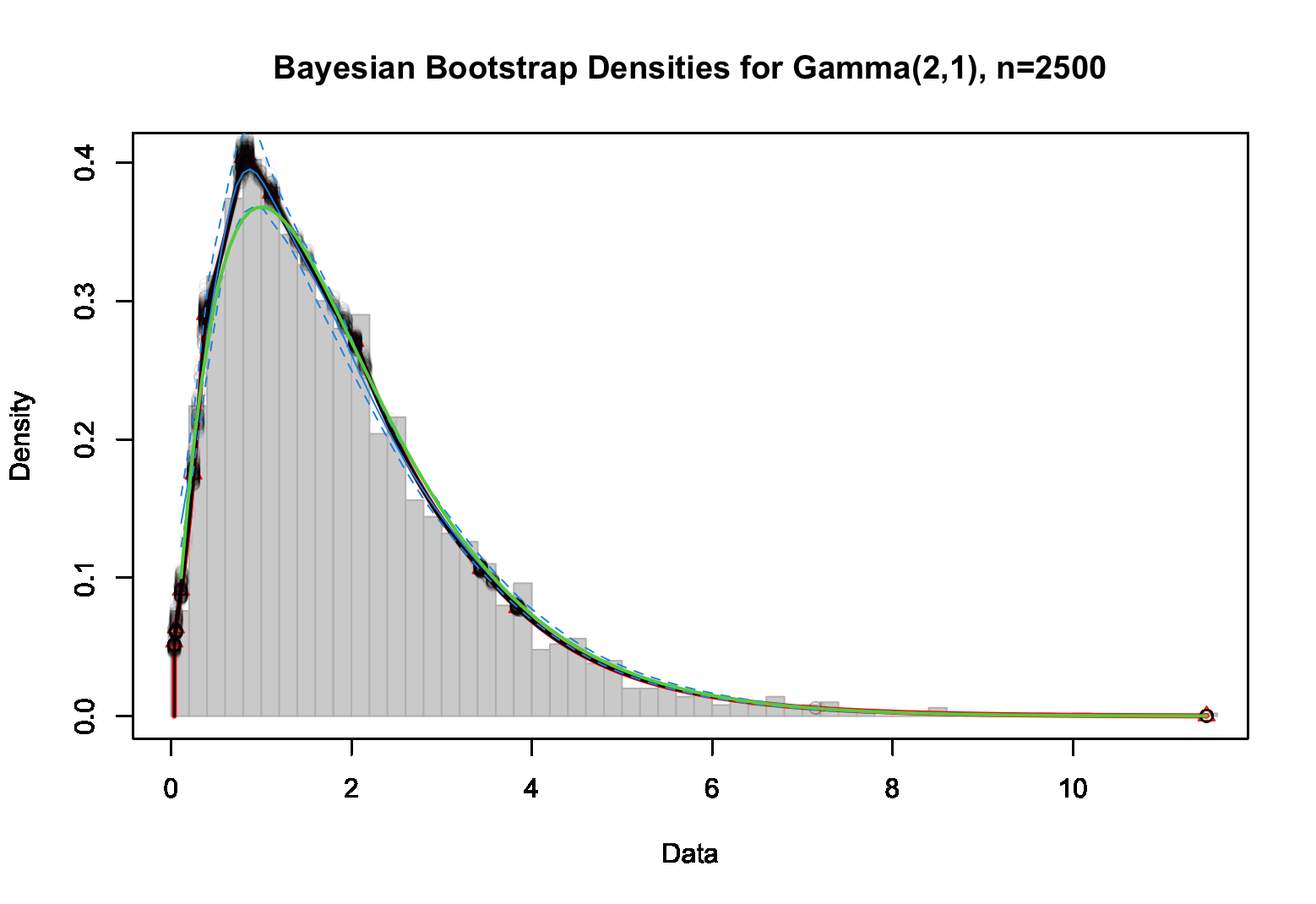}}
\caption{Histograms of sampled datasets, true densities (green), NPMLEs (red), NPMLE bootstrap samples (black), posterior means and pointwise credible bands using exponentiated Dirichlet process mixture priors (solid and dashed blue). Each plot shows 250 NPMLE bootstrap samples and 250 posteriors samples by MCMC. Each row represents the case with different initial sample sizes, which equals to 50, 200, 500, 2500 separately. Each column represents the case with different sampling distributions, which equals to $\mathcal{N}(0,0.1^2)$, $\mathcal{N}(100,0.1^2)$, and Gamma(2,1) respectively.}
\label{pl:Comp_S500}
\end{figure}

\begin{figure}[ht!]
\centering
\subfigure[]{
\label{pl:Comp_N0_S1000_n50}
\includegraphics[width=0.3\textwidth]{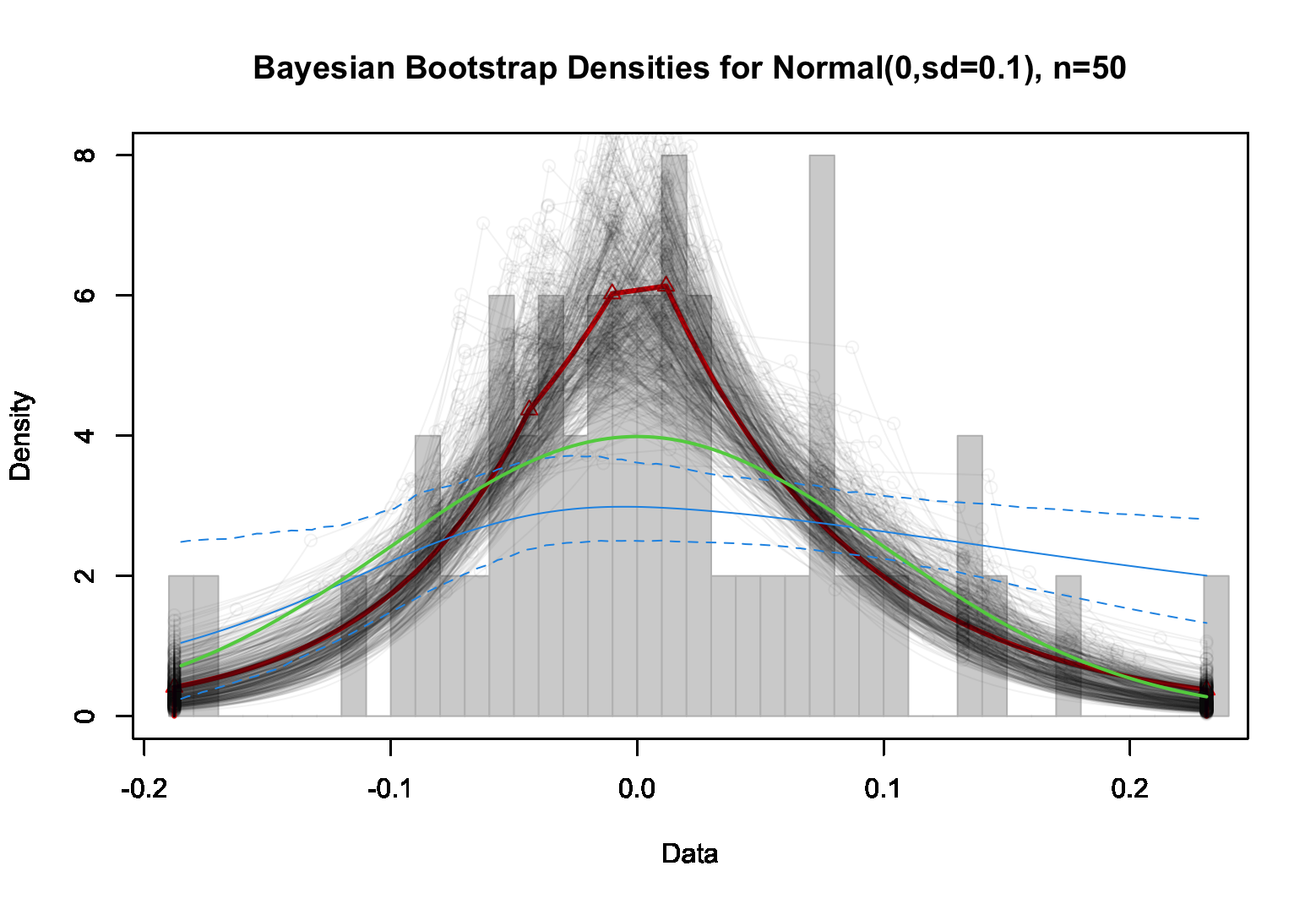}}
\subfigure[]{
\label{pl:Comp_N100_S1000_n50}
\includegraphics[width=0.3\textwidth]{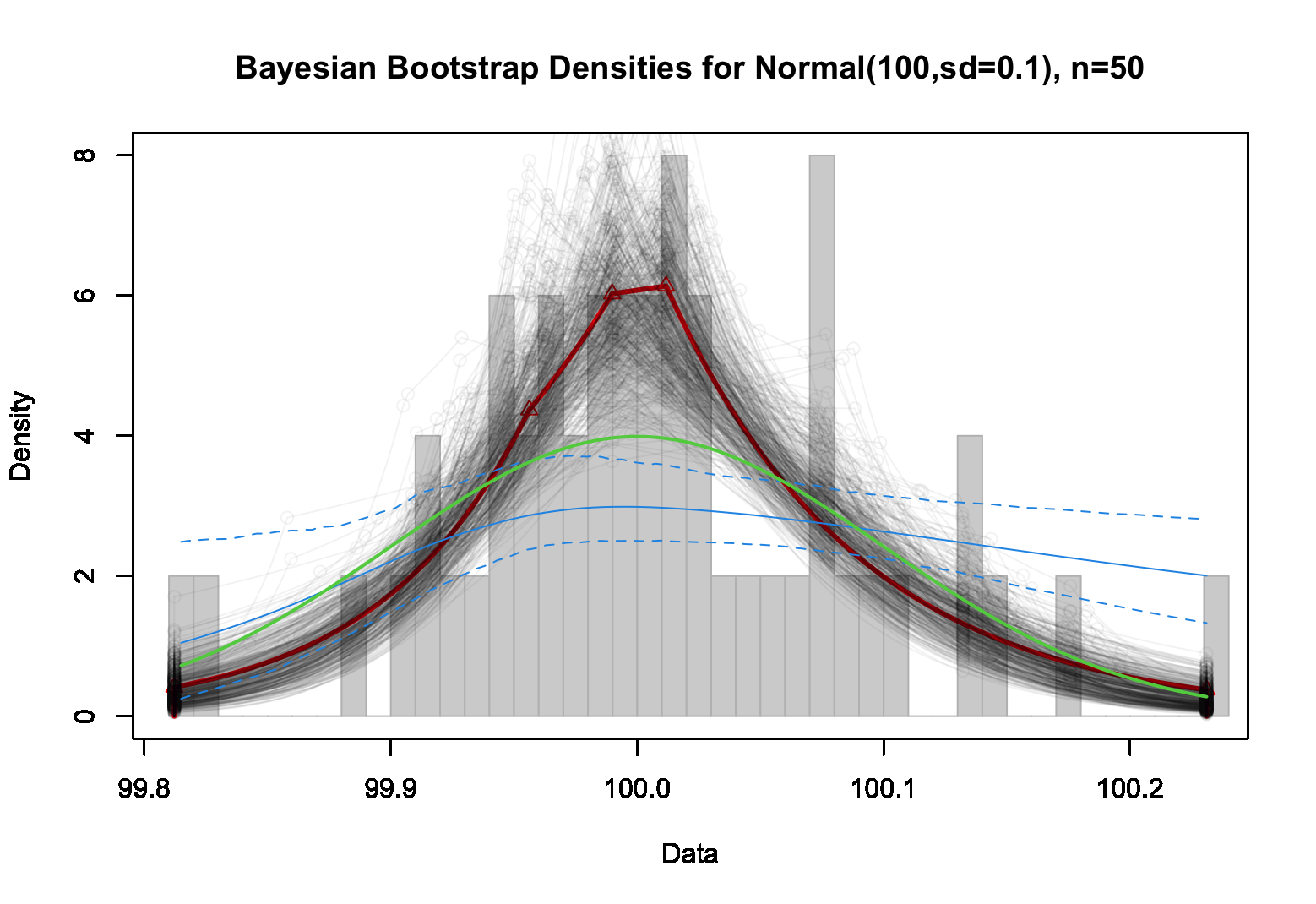}}
\subfigure[]{
\label{pl:Comp_G_S1000_n50}
\includegraphics[width=0.3\textwidth]{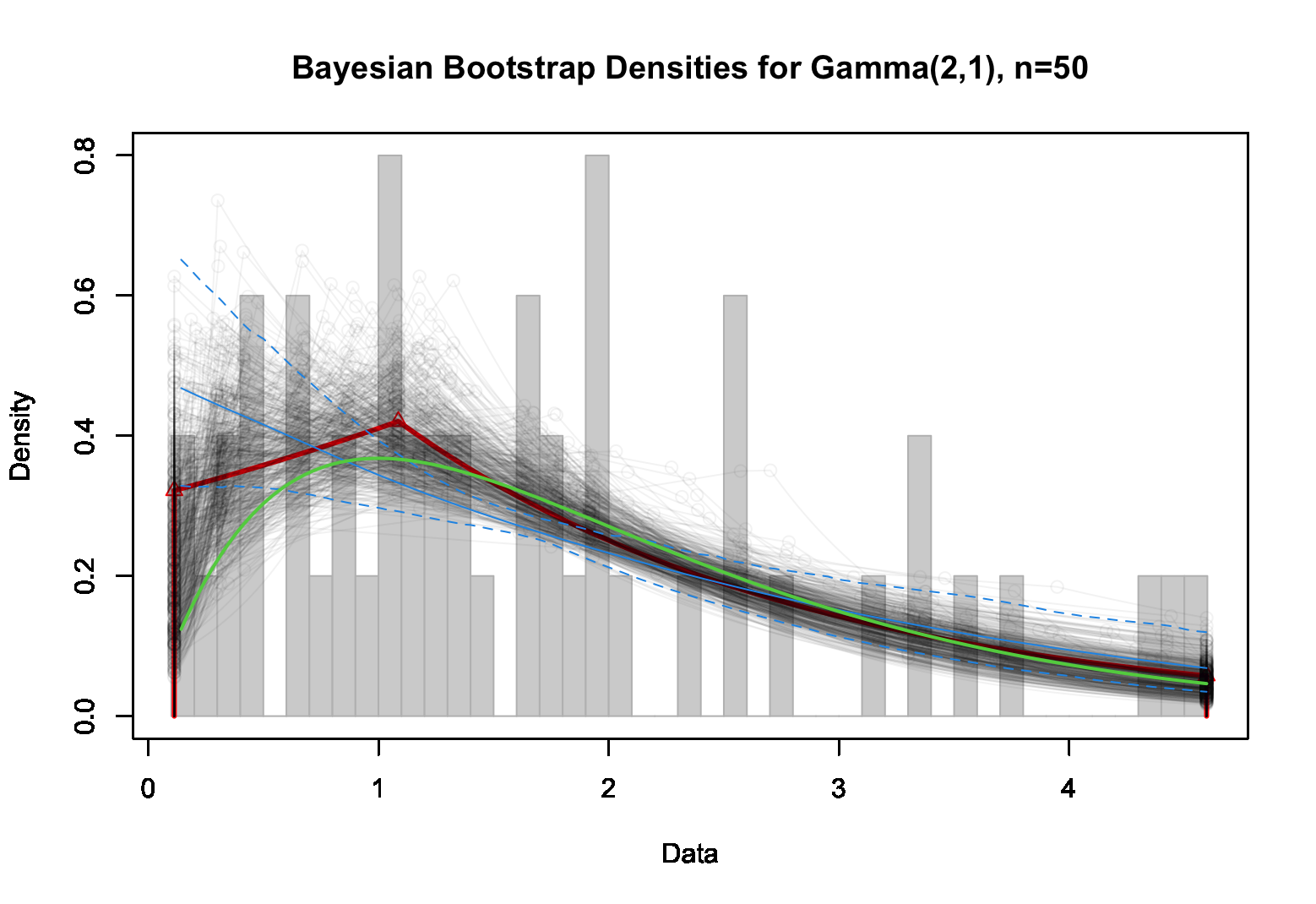}}
\subfigure[]{
\label{pl:Comp_N0_S1000_n200}
\includegraphics[width=0.3\textwidth]{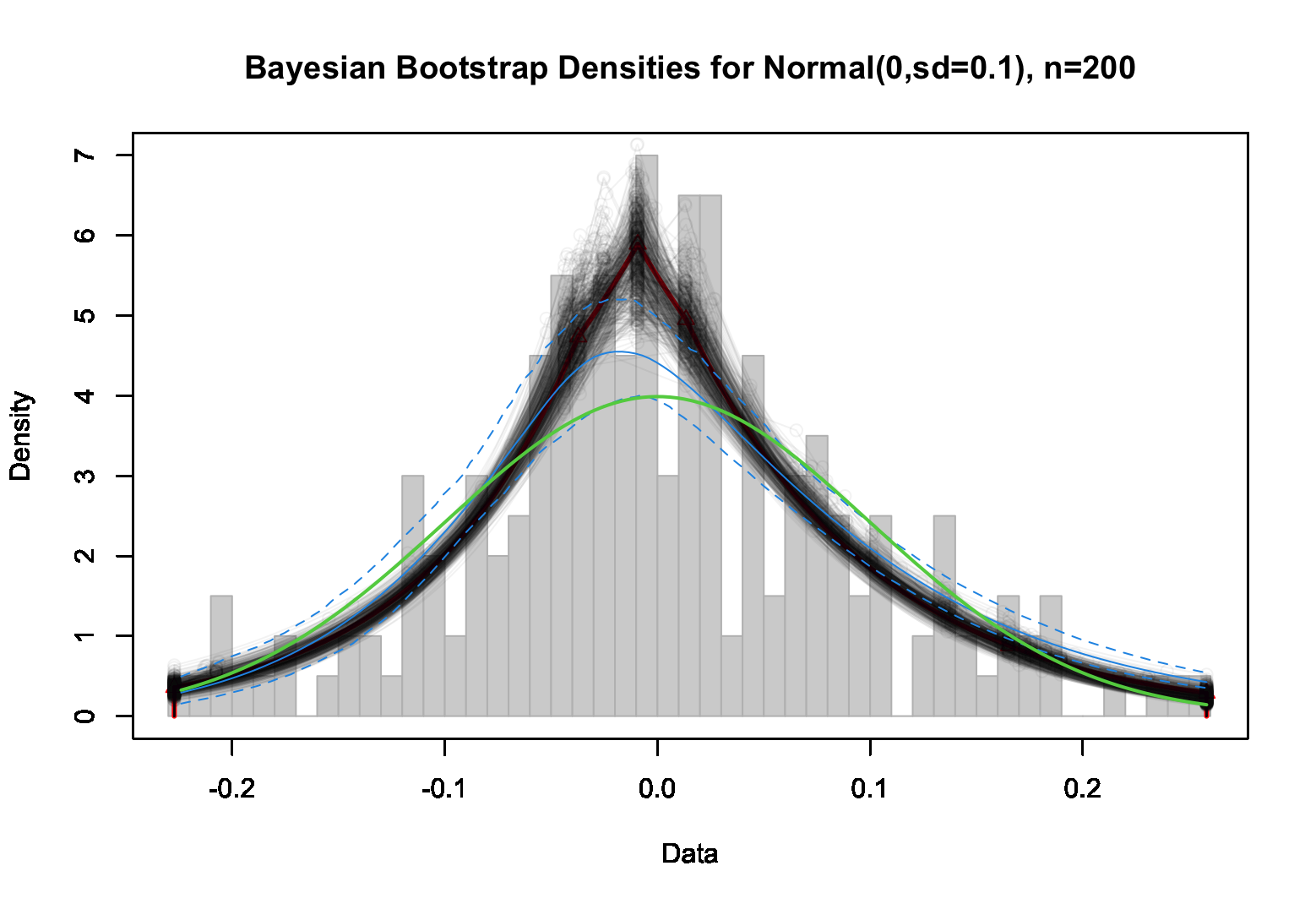}}
\subfigure[]{
\label{pl:Comp_N100_S1000_n200}
\includegraphics[width=0.3\textwidth]{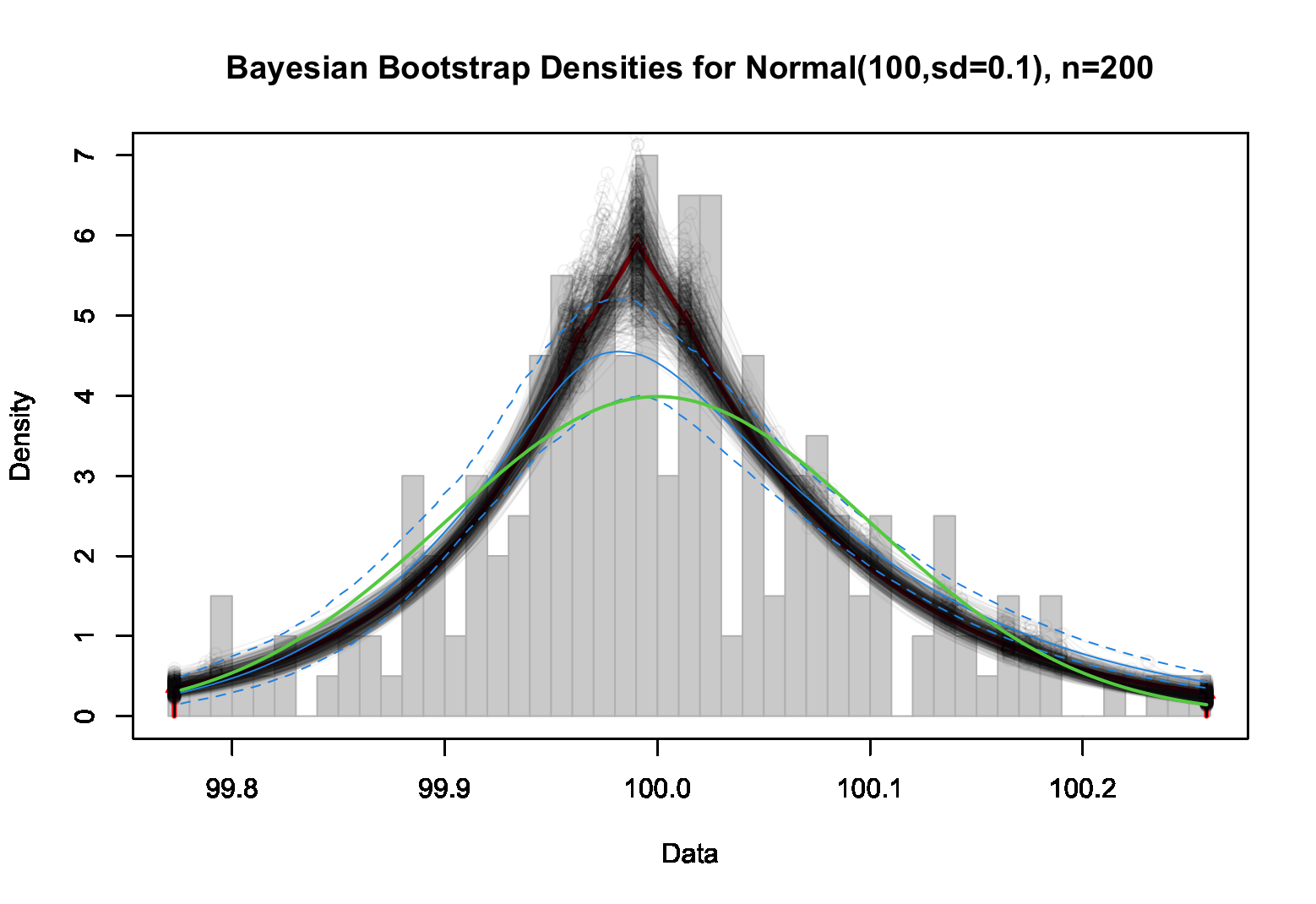}}
\subfigure[]{
\label{pl:Comp_G_S1000_n200}
\includegraphics[width=0.3\textwidth]{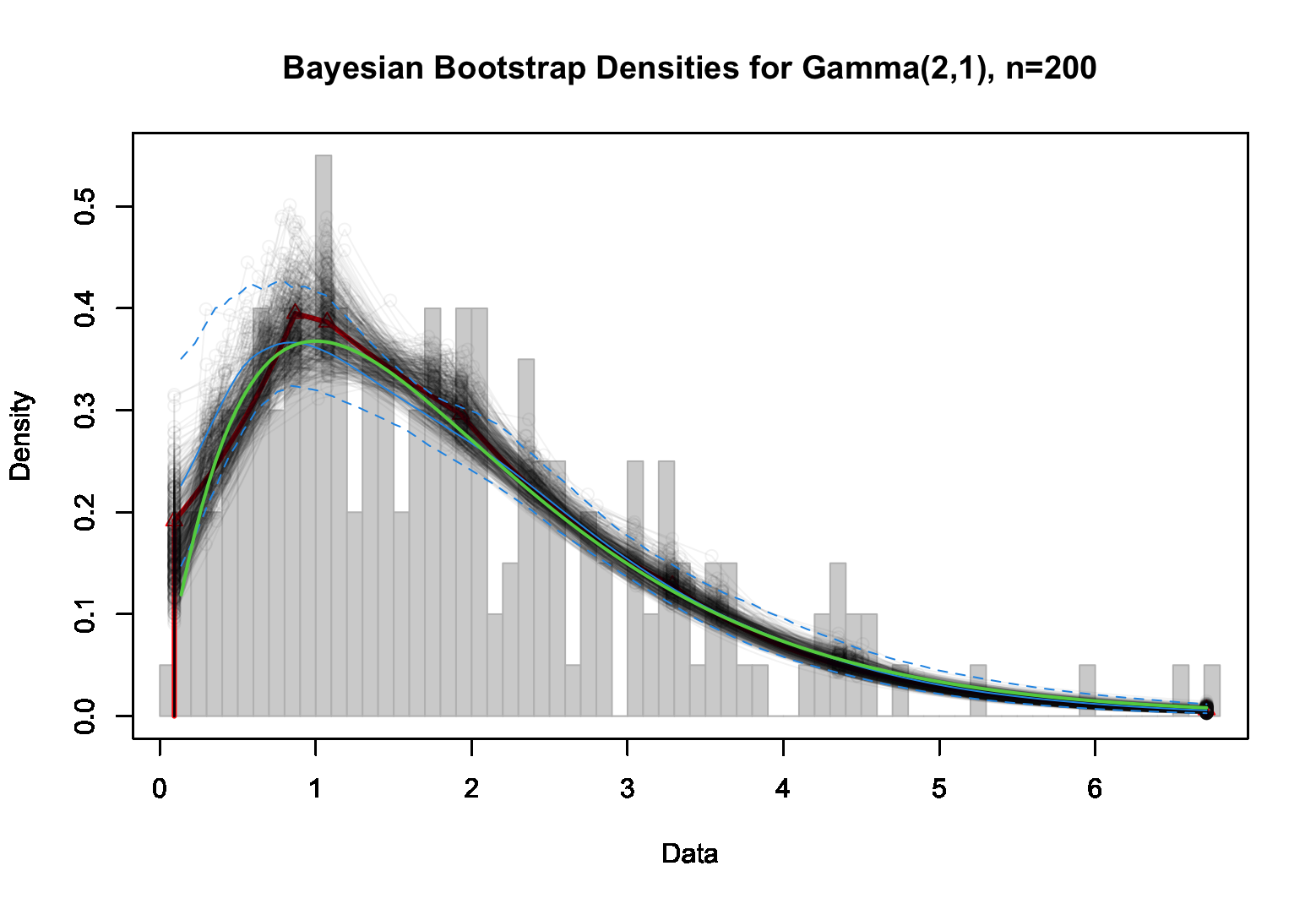}}
\subfigure[]{
\label{pl:Comp_N0_S1000_n500}
\includegraphics[width=0.3\textwidth]{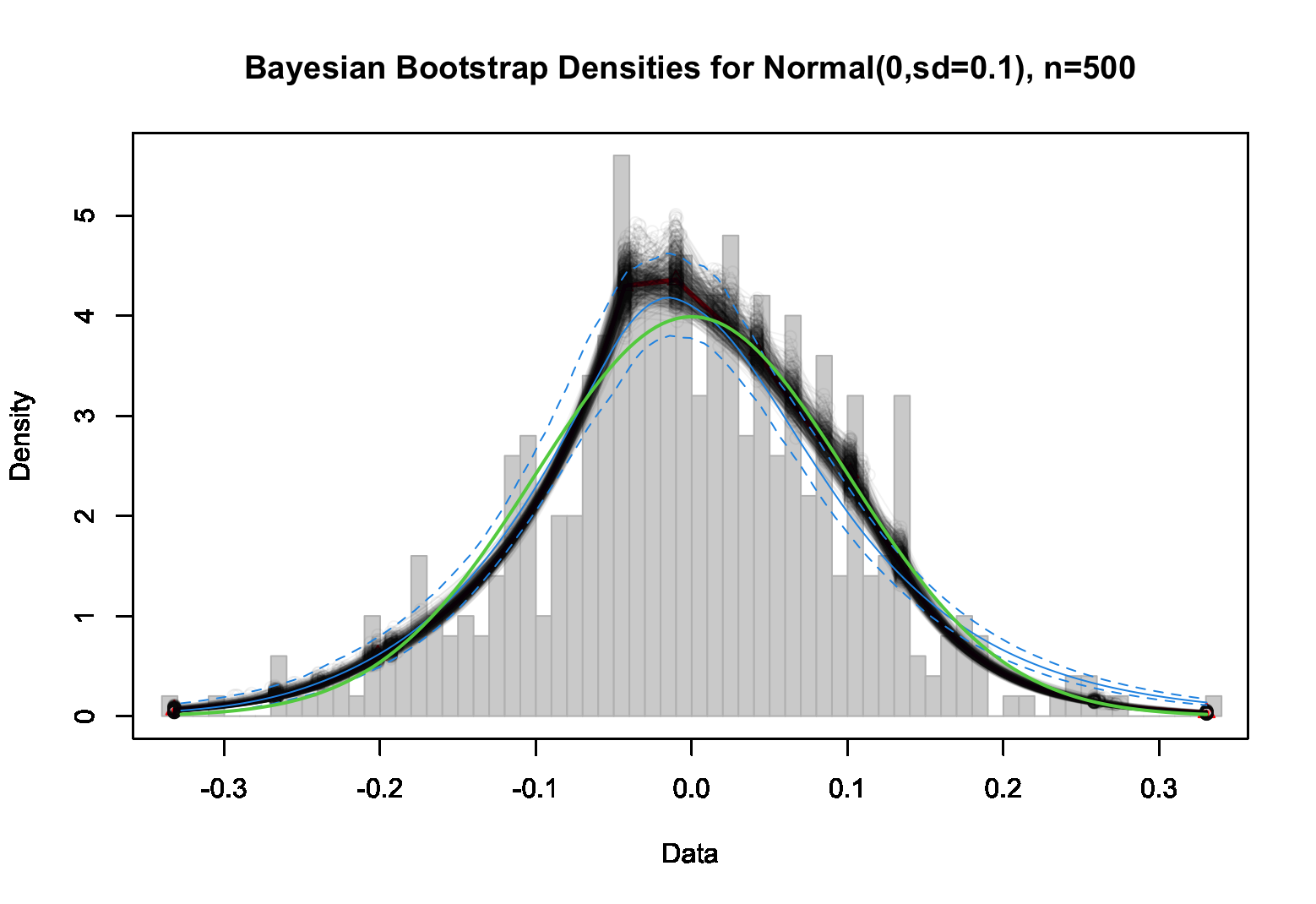}}
\subfigure[]{
\label{pl:Comp_N100_S1000_n500}
\includegraphics[width=0.3\textwidth]{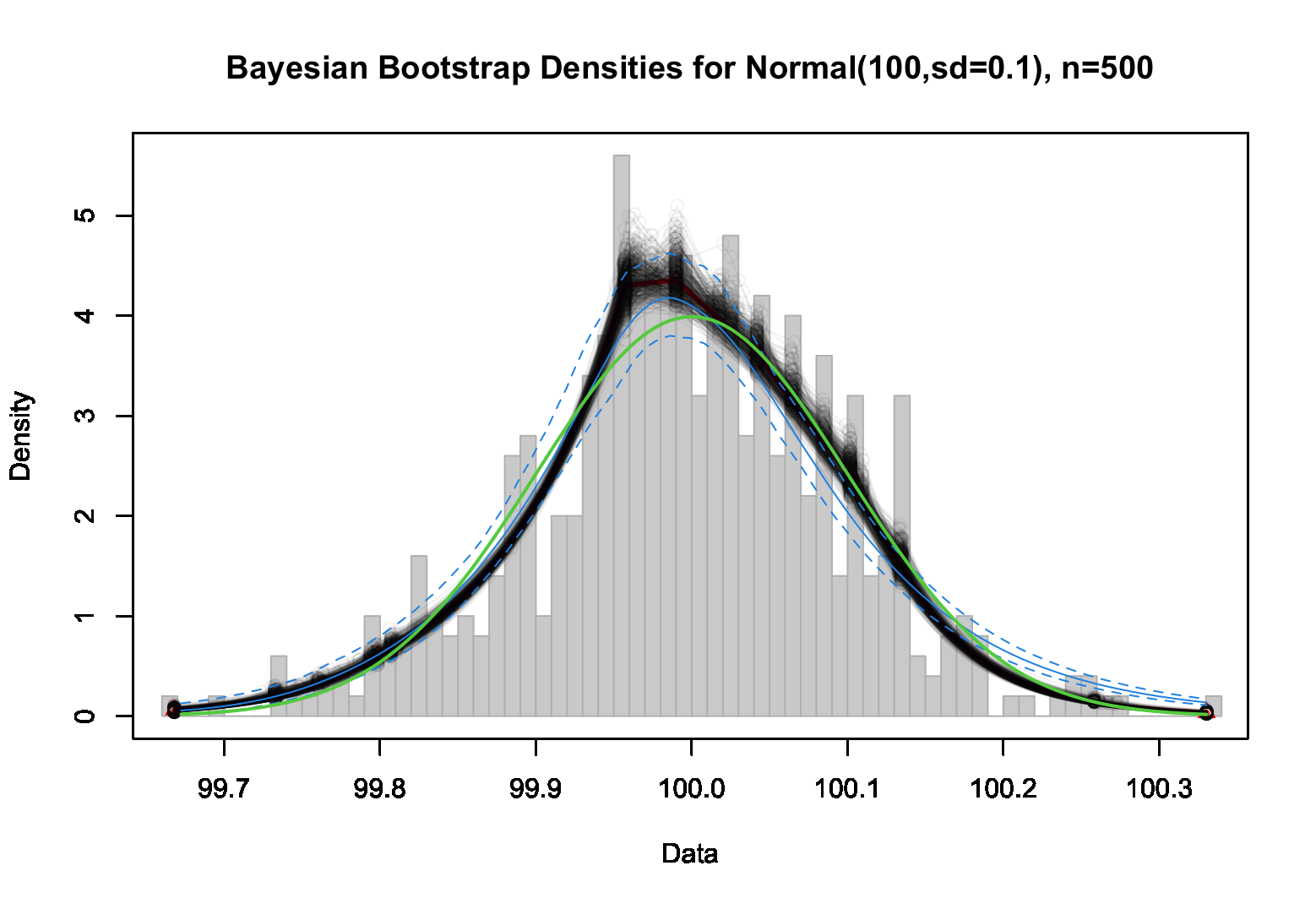}}
\subfigure[]{
\label{pl:Comp_G_S1000_n500}
\includegraphics[width=0.3\textwidth]{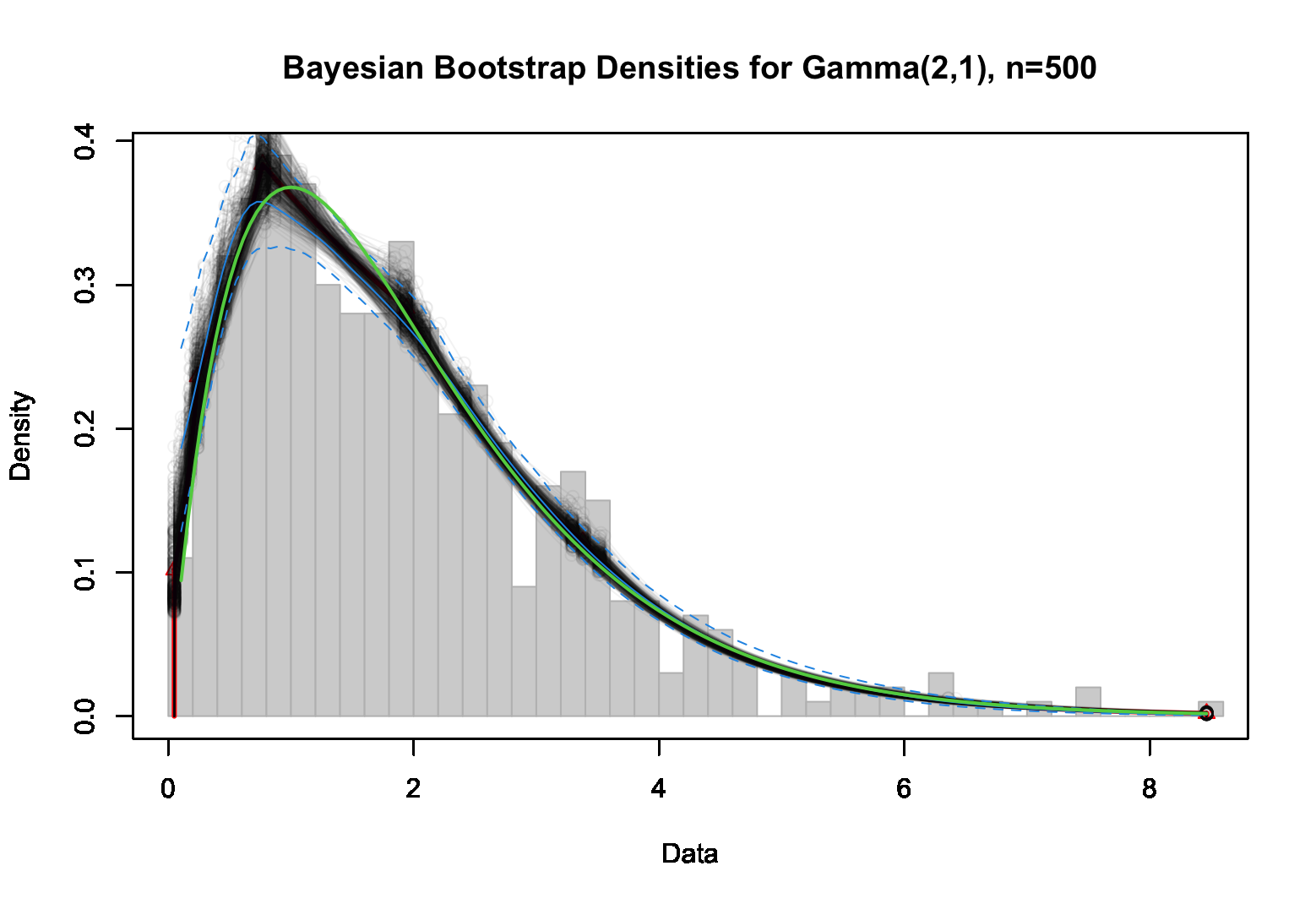}}
\subfigure[]{
\label{pl:Comp_N0_S1000_n2500}
\includegraphics[width=0.3\textwidth]{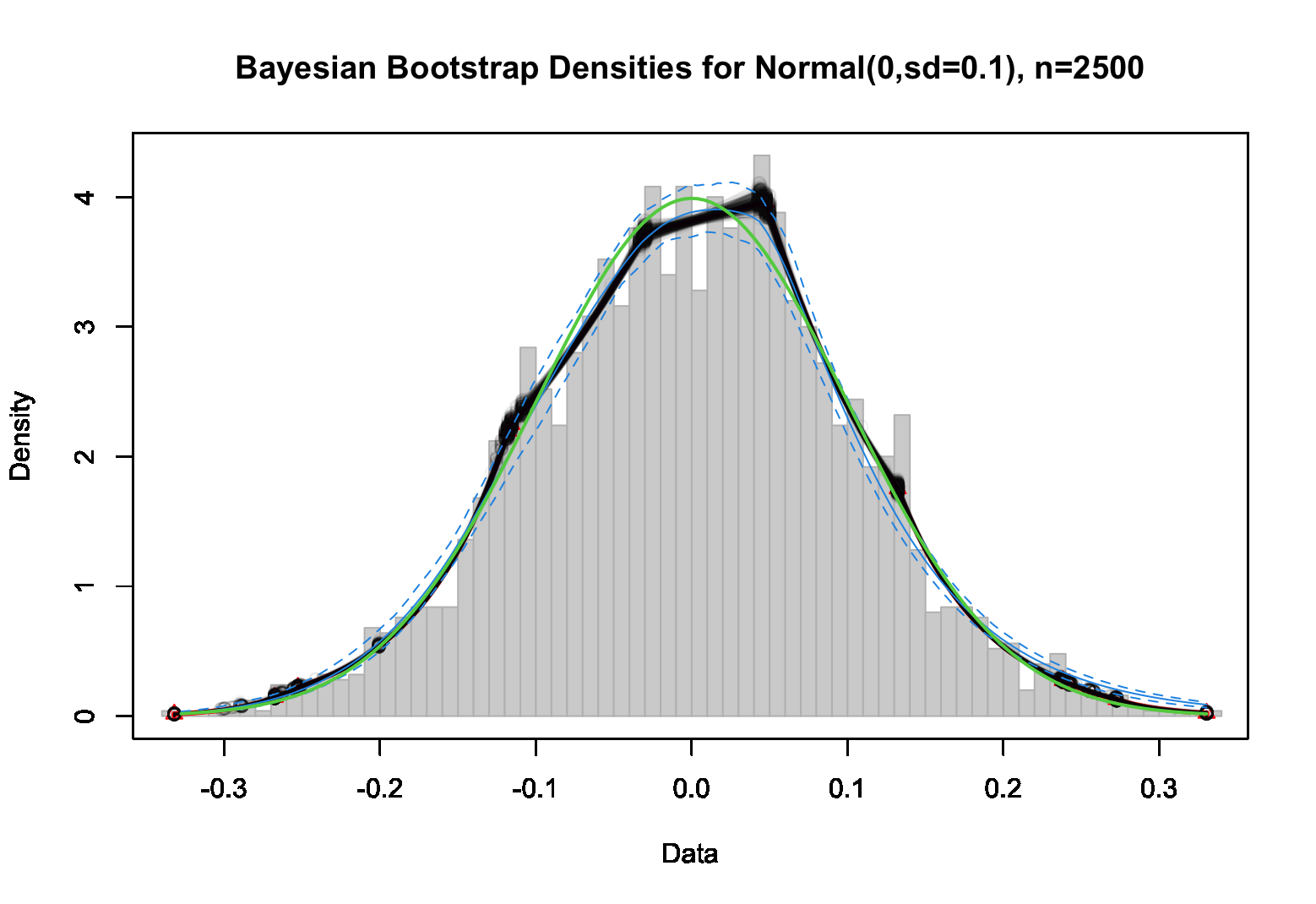}}
\subfigure[]{
\label{pl:Comp_N100_S1000_n2500}
\includegraphics[width=0.3\textwidth]{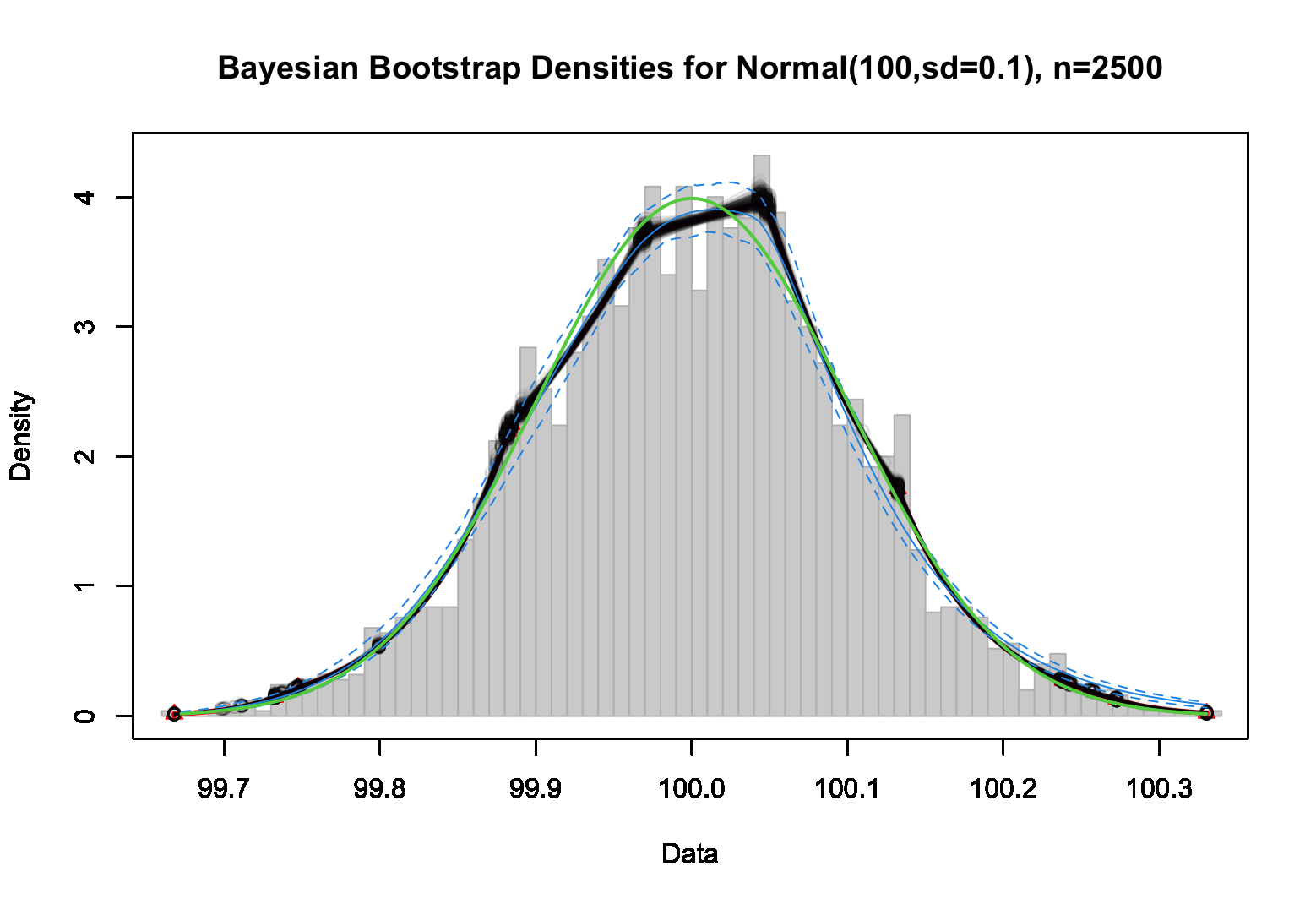}}
\subfigure[]{
\label{pl:Comp_G_S1000_n2500}
\includegraphics[width=0.3\textwidth]{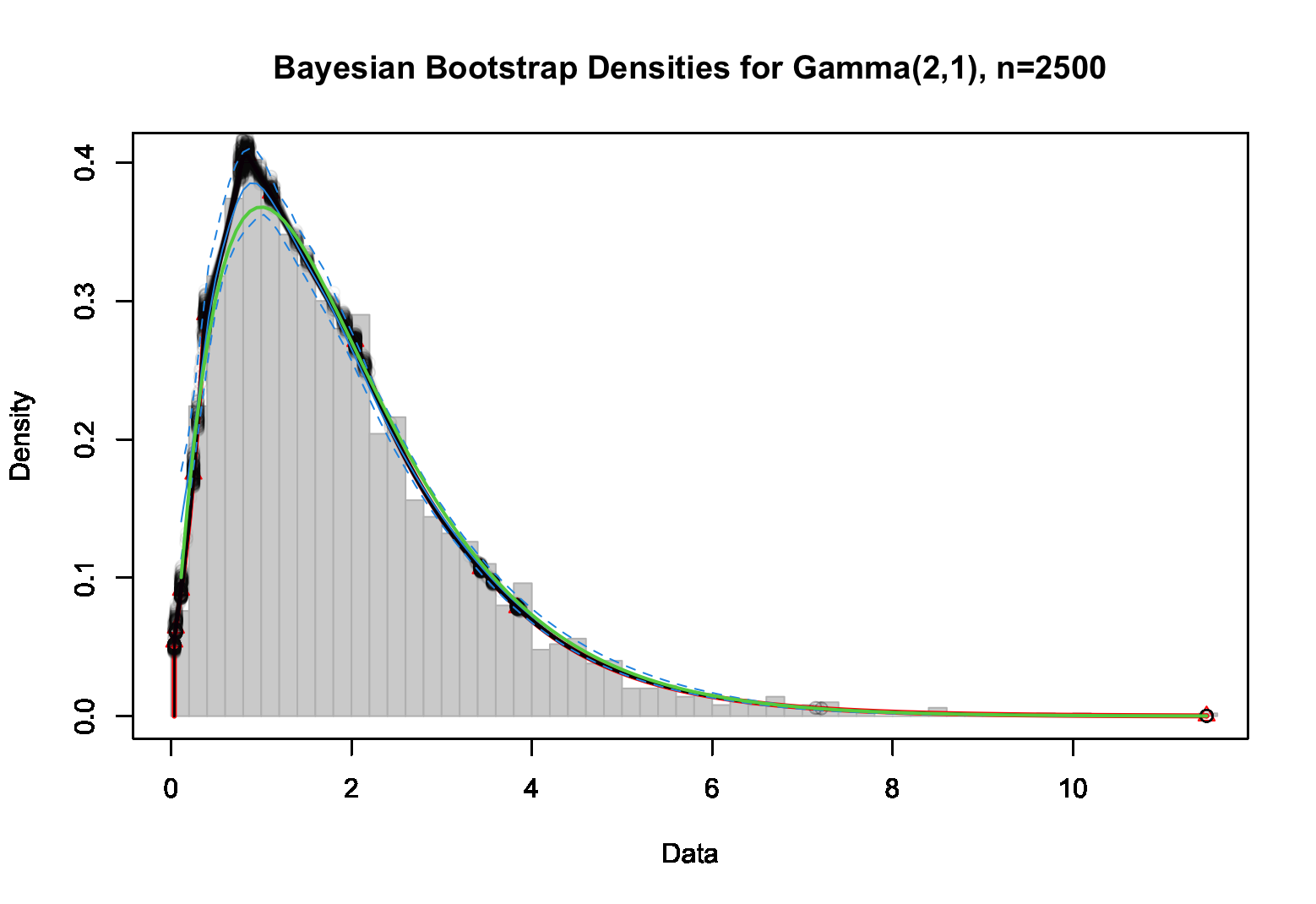}}
\caption{Histograms of sampled datasets, true densities (green), NPMLEs (red), NPMLE bootstrap samples (black), posterior means and pointwise credible bands using exponentiated Dirichlet process mixture priors (solid and dashed blue). Each plot shows 500 NPMLE bootstrap samples and 500 posteriors samples by MCMC. Each row represents the case with different initial sample sizes, which equals to 50, 200, 500, 2500 separately. Each column represents the case with different sampling distributions, which equals to $\mathcal{N}(0,0.1^2)$, $\mathcal{N}(100,0.1^2)$, and Gamma(2,1) respectively.}
\label{pl:Comp_S1000}
\end{figure}

\begin{figure}[ht!]
\centering
\subfigure[]{
\label{pl:Comp_N0_S10000_n50}
\includegraphics[width=0.3\textwidth]{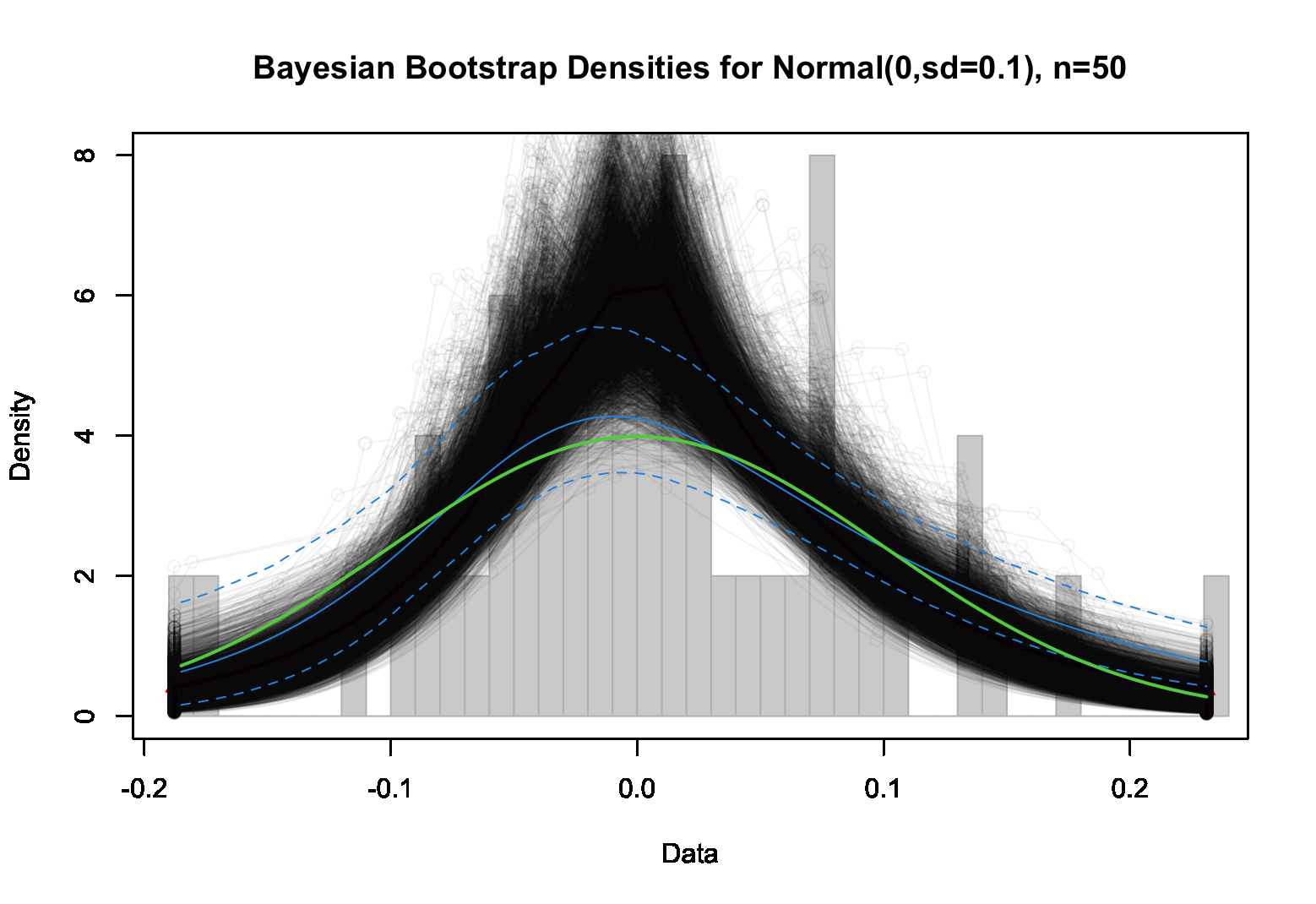}}
\subfigure[]{
\label{pl:Comp_N100_S10000_n50}
\includegraphics[width=0.3\textwidth]{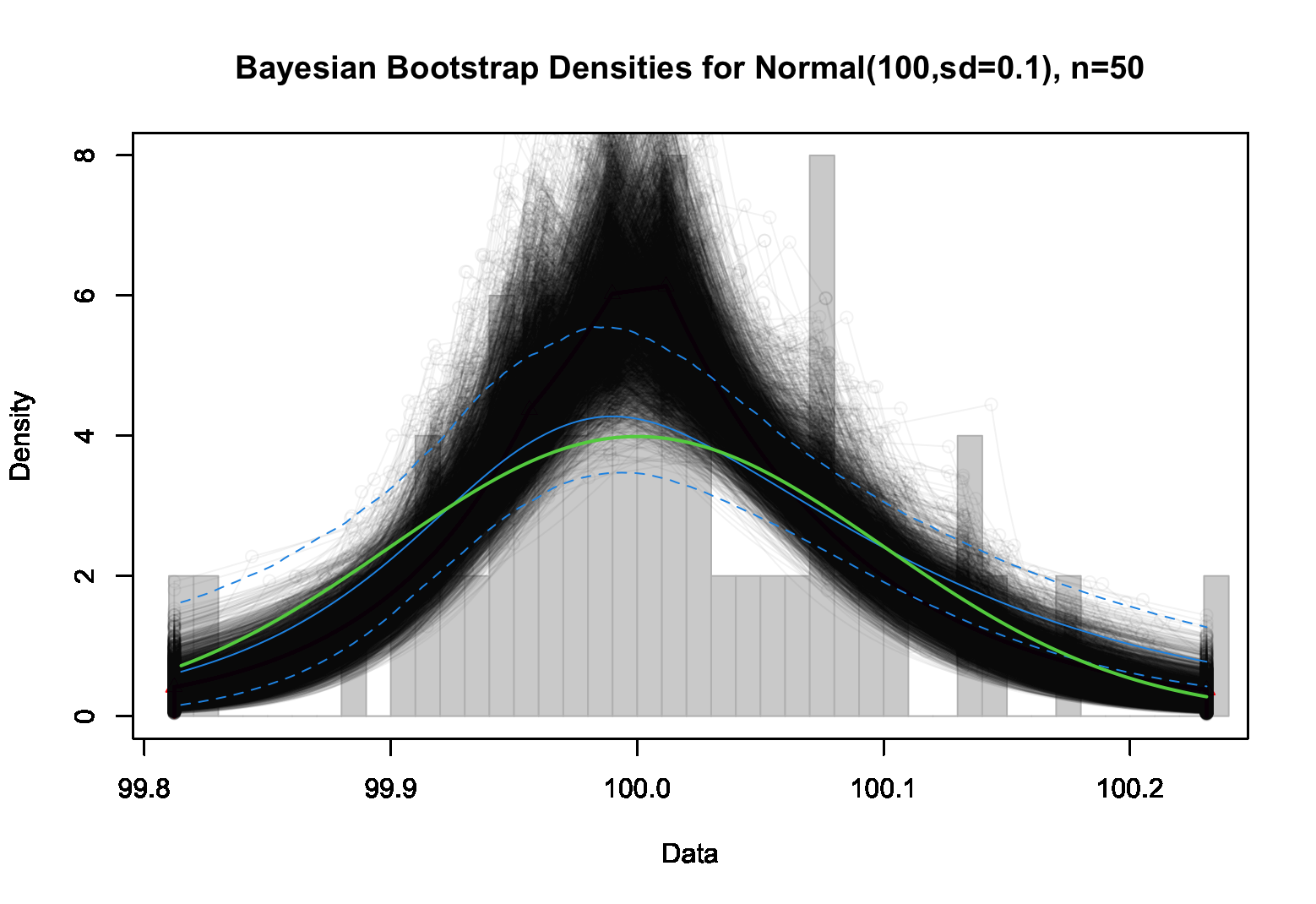}}
\subfigure[]{
\label{pl:Comp_G_S10000_n50}
\includegraphics[width=0.3\textwidth]{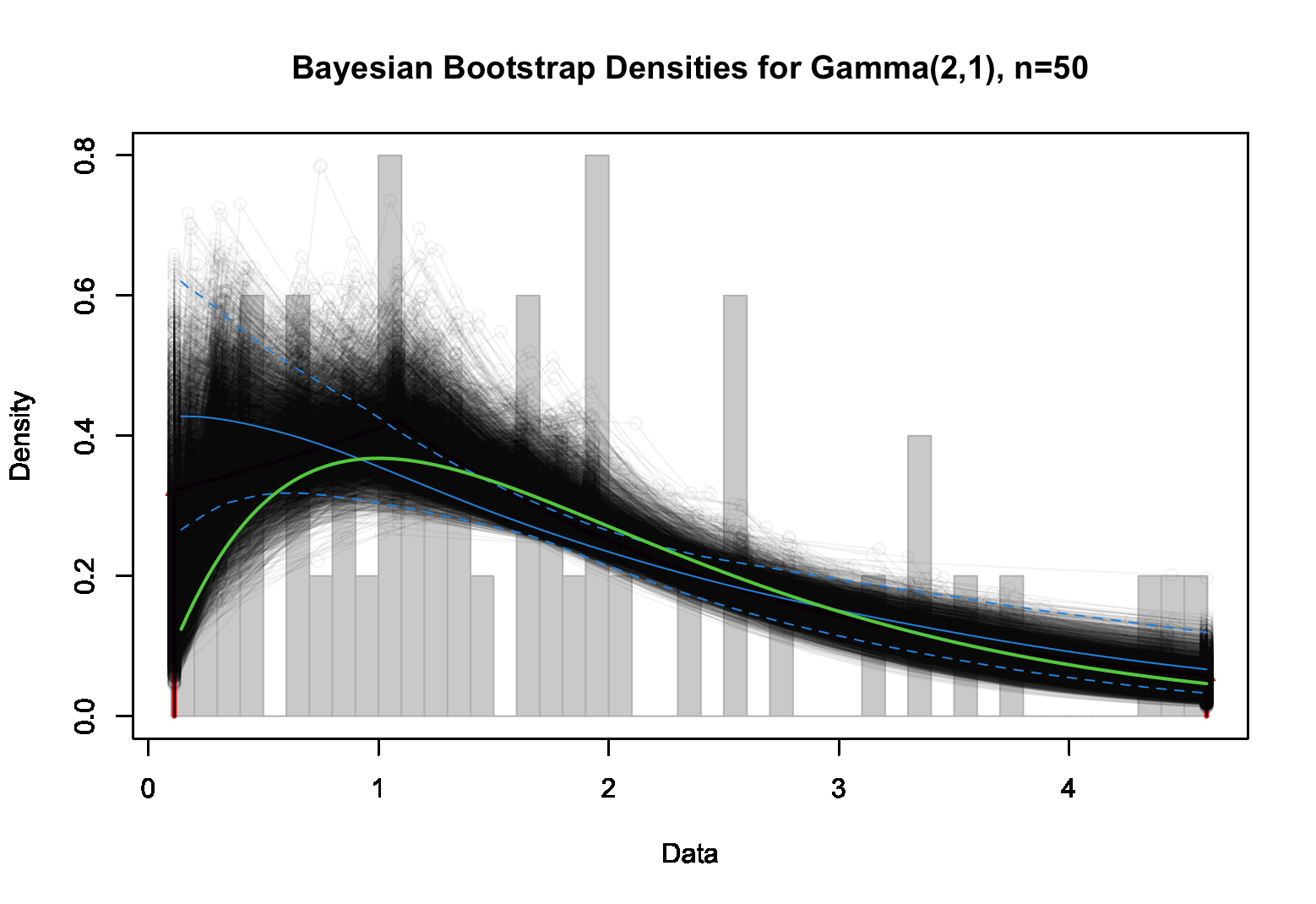}}
\subfigure[]{
\label{pl:Comp_N0_S10000_n200}
\includegraphics[width=0.3\textwidth]{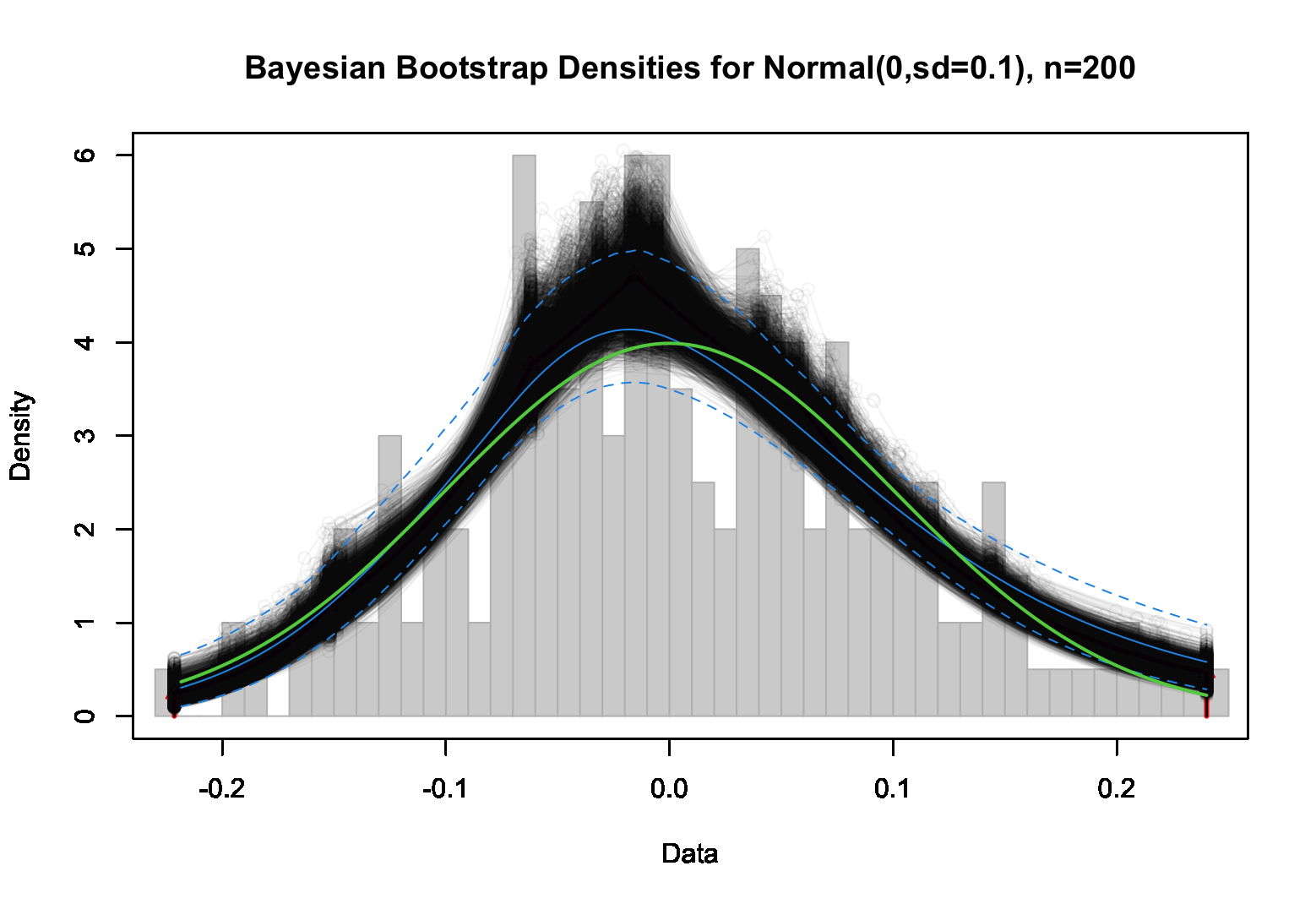}}
\subfigure[]{
\label{pl:Comp_N100_S10000_n200}
\includegraphics[width=0.3\textwidth]{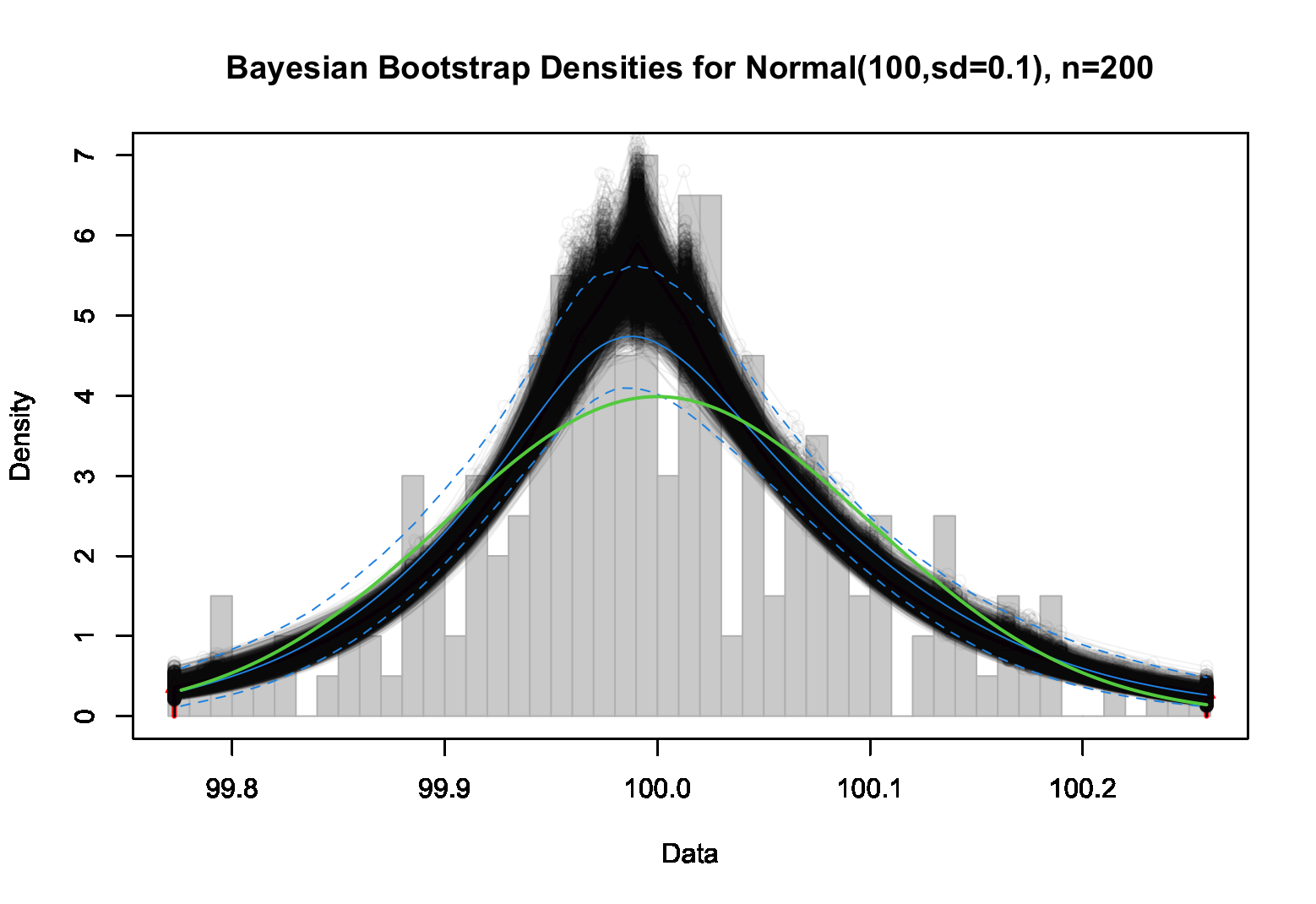}}
\subfigure[]{
\label{pl:Comp_G_S10000_n200}
\includegraphics[width=0.3\textwidth]{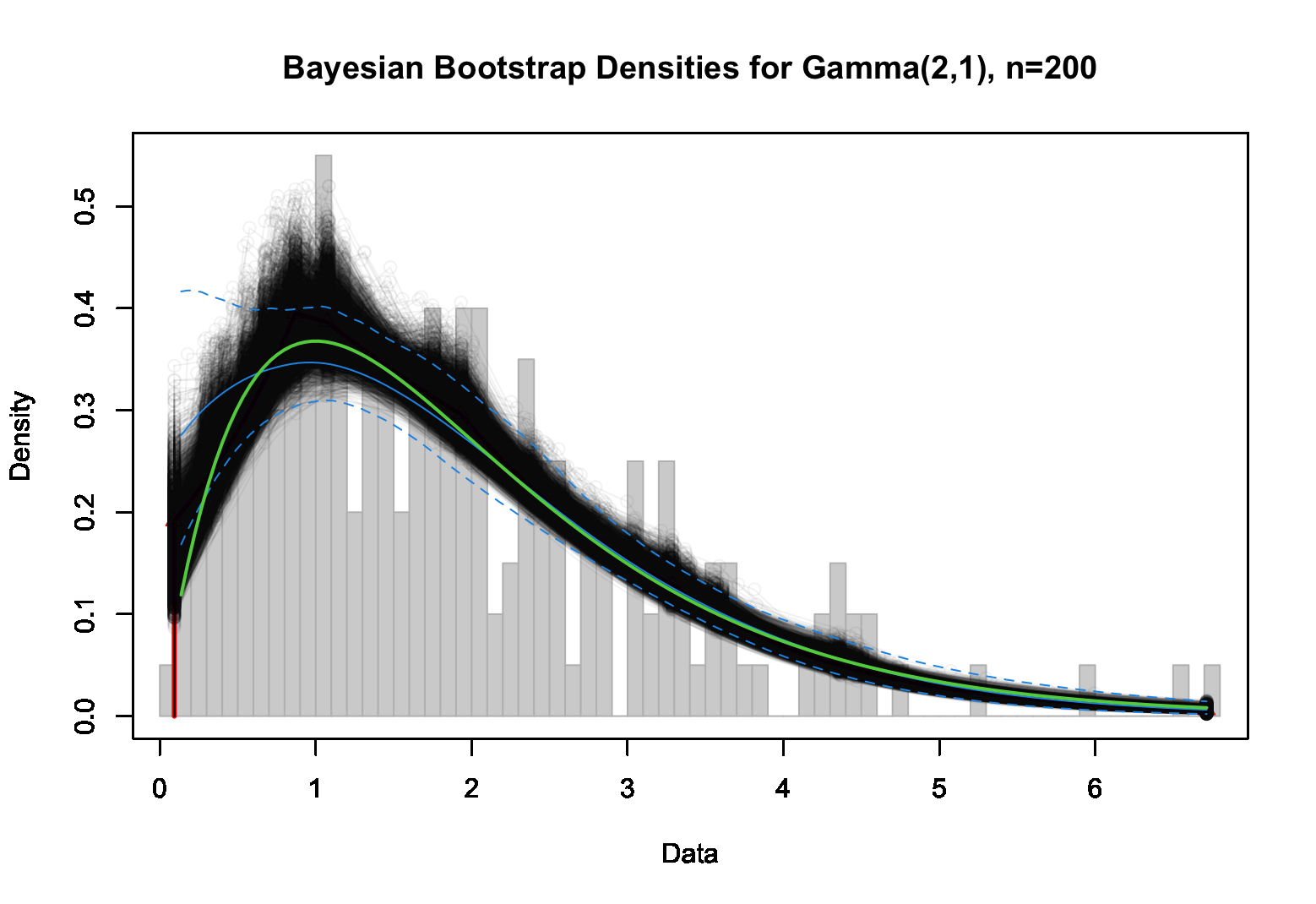}}
\subfigure[]{
\label{pl:Comp_N0_S10000_n500}
\includegraphics[width=0.3\textwidth]{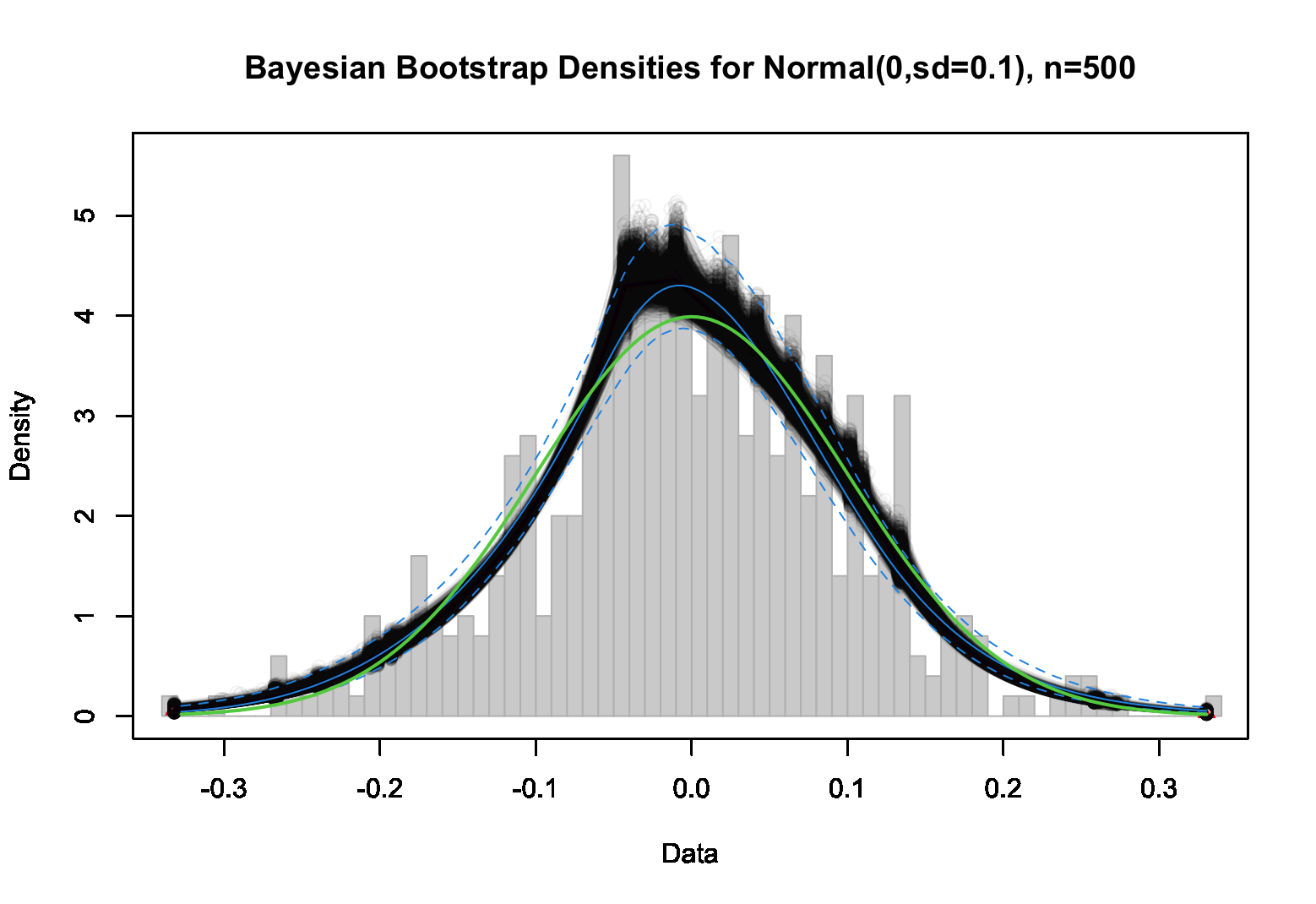}}
\subfigure[]{
\label{pl:Comp_N100_S10000_n500}
\includegraphics[width=0.3\textwidth]{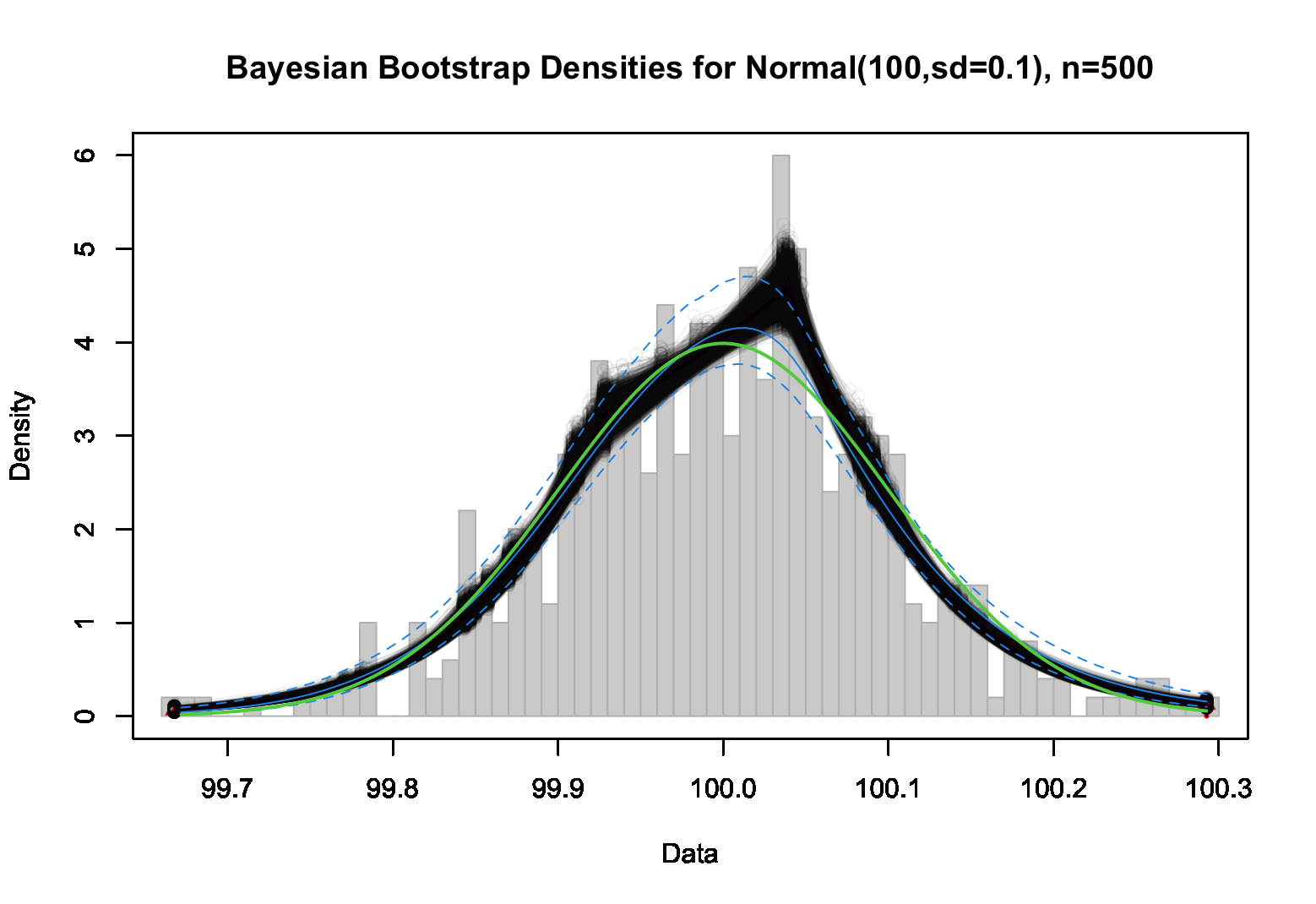}}
\subfigure[]{
\label{pl:Comp_G_S10000_n500}
\includegraphics[width=0.3\textwidth]{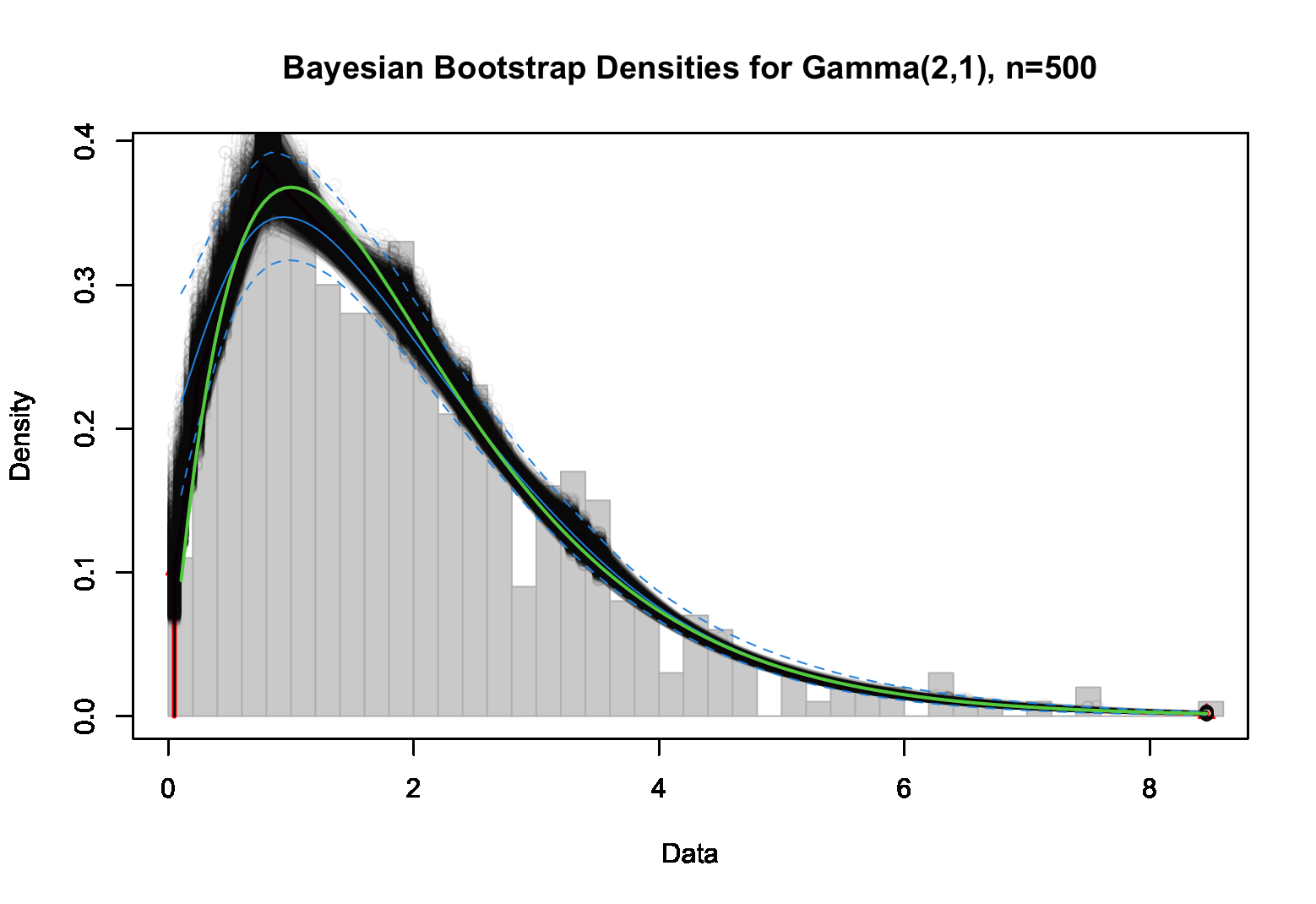}}
\subfigure[]{
\label{pl:Comp_N0_S10000_n2500}
\includegraphics[width=0.3\textwidth]{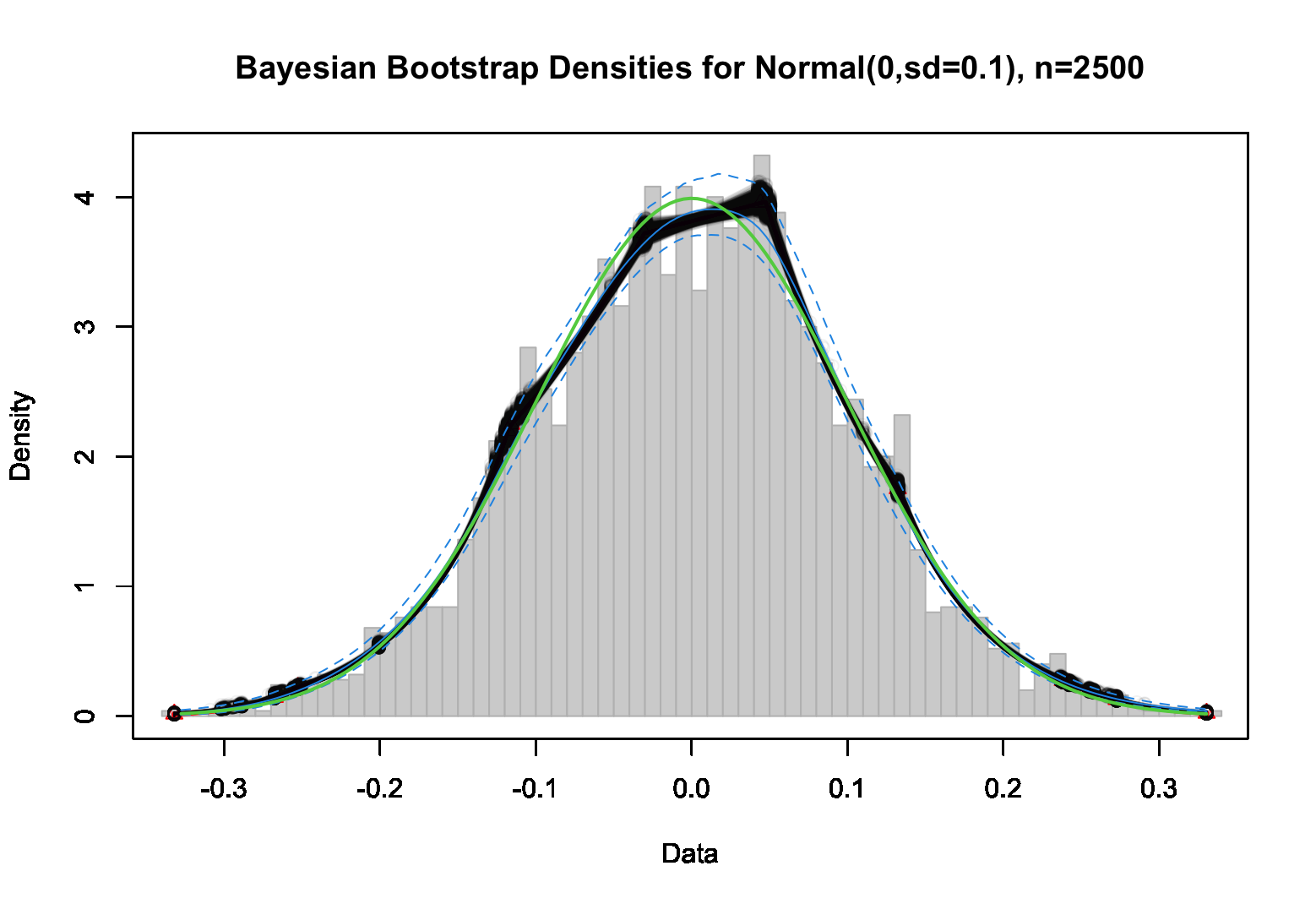}}
\subfigure[]{
\label{pl:Comp_N100_S10000_n2500}
\includegraphics[width=0.3\textwidth]{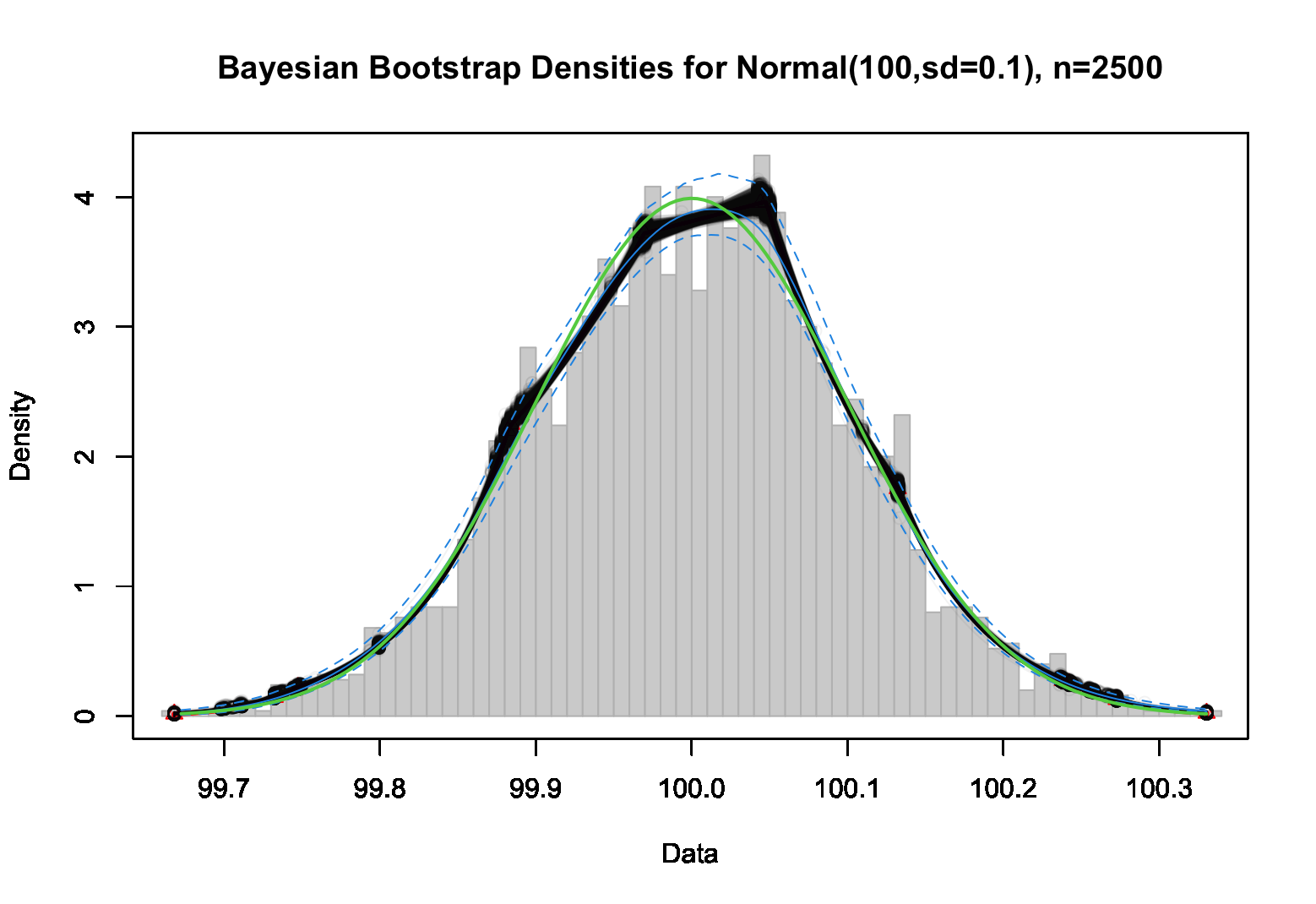}}
\subfigure[]{
\label{pl:Comp_G_S10000_n2500}
\includegraphics[width=0.3\textwidth]{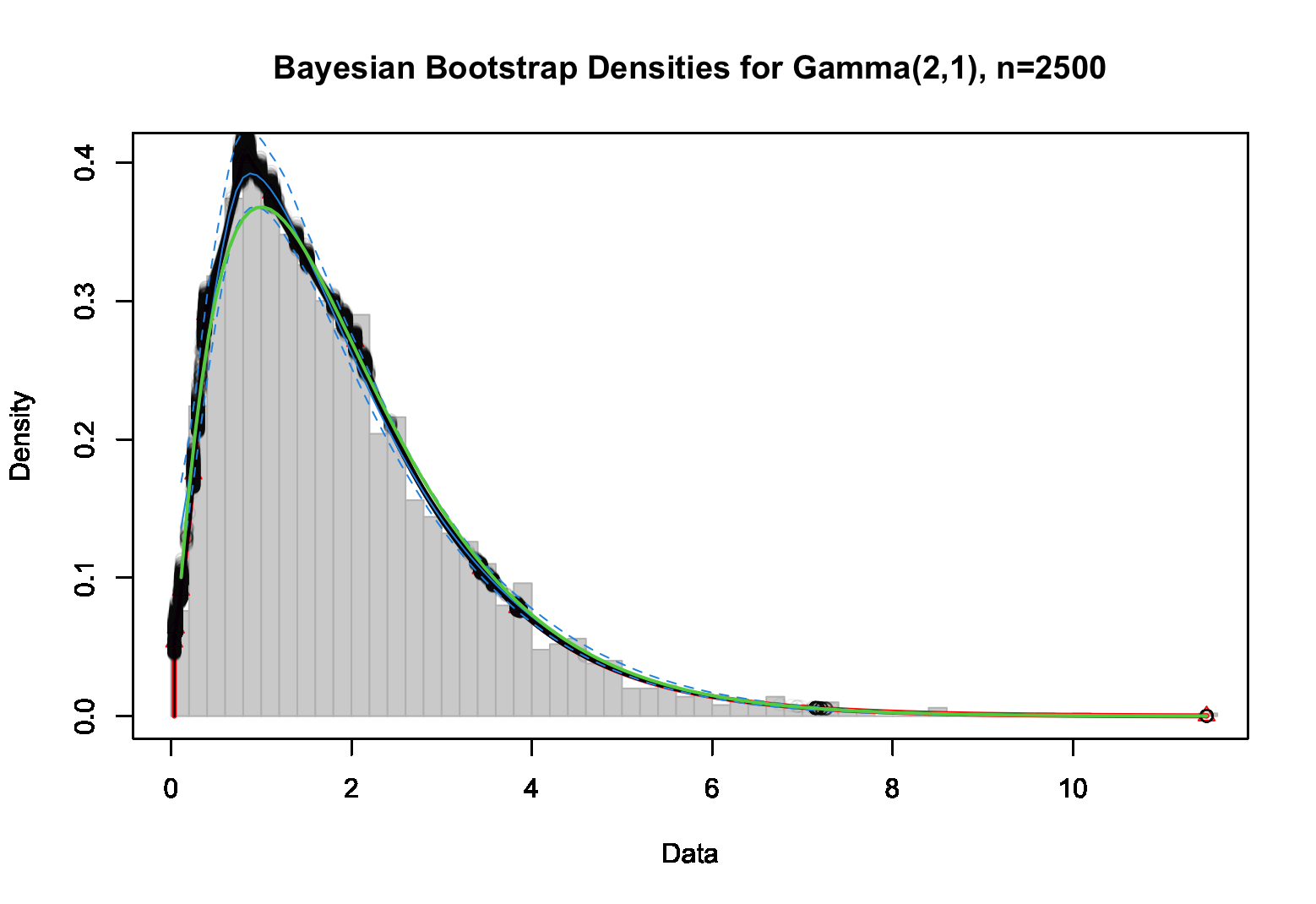}}
\caption{Histograms of sampled datasets, true densities (green), NPMLEs (red), NPMLE bootstrap samples (black), posterior means and pointwise credible bands using exponentiated Dirichlet process mixture priors (solid and dashed blue). Each plot shows 5000 NPMLE bootstrap samples and 5000 posteriors samples by MCMC. Each row represents the case with different initial sample sizes, which equals to 50, 200, 500, 2500 separately. Each column represents the case with different sampling distributions, which equals to $\mathcal{N}(0,0.1^2)$, $\mathcal{N}(100,0.1^2)$, and Gamma(2,1) respectively.}
\label{pl:Comp_S10000}
\end{figure}

\subsection{Comparison with Exponentiated Dirichlet Process Mixture Prior}
\label{subsec:Comparison}
Here we compare our approach with that of \cite{Mariucci_2020}. These authors use a Bayesian nonparametric approach to estimate the log concave density functions. Further, they also provide uncertainty quantification of the estimate, using an exponentiated Dirichlet process mixture prior. As mentioned in Section~\ref{sec:Introduction}, \cite{Mariucci_2020} uses
$$
w(x) = \gamma_1\int_{0}^{\infty}\frac{u\wedge x}{u}dP(u) - \gamma_2x
$$
where $\gamma_1>0, \gamma_2\in\mathbb{R}$ and $P$ is a probability measure on $[0,\infty)$ as a approximated representative of concave functions, which can be arbitrarily closed to any concave functions in Hellinger distance. Then \cite{Mariucci_2020} uses a similar form for the prior of the log-density $w: [a_n,b_n]\to \mathbb{R}$:
$$
W(x) = \gamma_1\int_{0}^{b_n-a_n}\frac{u\wedge (x-a_n)}{u}dP(u) - \gamma_2(x-a_n).
$$
In particular, they use the Dirichlet process prior for $P$. It is now possible to sample from the posterior via Markov Chain Monte Carlo (MCMC).

The  comparisons are made in Figs~\ref{pl:Comp_S500}, \ref{pl:Comp_S1000} and \ref{pl:Comp_S10000}. The first point to make is that \cite{Mariucci_2020} requires a prior distribution while our approach only requires the NPMLE, just as with the original BB. 
%If not appropriately selected, when the original data set is small, the improper prior will influence the posterior deeply. 
Further, \cite{Mariucci_2020} require MCMC to obtain posterior samples, which can only be implemented sequentially, and not in parallel. %In some circumstances, it will be very time-consuming, while our method can work in parallell. However, due to the sub-martingale properties, the mean of the nonparametric posterior using NPMLE bootstrap is not unbiased. We can see more from the experiment results.

When the original sample size becomes large, \cite{Mariucci_2020}'s method becomes much slower.  In our experiments, we use 11 cores to do parallel computing for each of the martingales. Hence, it is no surprise that our method is in general faster running than that of \cite{Mariucci_2020}.

%\newpage

\section{Summary and discussion}
\label{sec:Summary_and_Discussion}
In this paper we have presented a martingale posterior distribution for log-concave density functions, using the same idea of sampling and updating as is done with the Bayesian bootstrap. Starting  with the log-concave NPMLE, we construct a version of the martingale posterior distribution, which can be directly compared with the Bayesian bootstrap, which starts off with the empirical distribution function. 

The sampling algorithm can be implemented in parallel, and so it is fast to run. We have also presented the corresponding theory demonstrating that the posterior is well defined. Future work will consider alternative types of shape constrained density or distribution functions. 

%\newpage

%%%%%%%%%%%%%%%%%%%%%%%%%%%%%%%%%%%%%%%%%%%%%%%%%%%%%%%%%%%%%%%%%%%%%%%%%%
\bibliographystyle{abbrv}
\bibliography{Fuheng}
\end{document}